\documentclass[nonatbib,nocopyrightspace,reprint]{sigplanconf}

\usepackage{ifthen}
\newboolean{uglytwosided}
\setboolean{uglytwosided}{true}

\usepackage[latin2]{inputenc}
\usepackage[english]{babel}
\usepackage{amsmath}

\usepackage{todonotes}
\usepackage{amsthm}
\usepackage{amssymb}
\usepackage{microtype}
\usepackage{xspace}
\usepackage{comment}
\usepackage{letltxmacro}
\usepackage{comment}
\usepackage{enumerate}
\usepackage{tikz}
\usetikzlibrary{calc}
\usepackage{url}
\usepackage{eufrak}
\usepackage{stmaryrd}

\usepackage{graphicx}
\usepackage{caption}
\usepackage{subcaption}
\usepackage{float}

\renewcommand{\topfraction}{0.85}

\renewcommand{\floatpagefraction}{0.75}

\newtheorem{theorem}{Theorem}[section]
\newtheorem{aplemma}[theorem]{{\color{red}Lemma}}
\newtheorem{lemma}[theorem]{Lemma}
\newtheorem{claim}{Claim}

\theoremstyle{definition}
\newtheorem{definition}[theorem]{Definition}
\newtheorem{example}[theorem]{Example}

\def\cqedsymbol{\ifmmode$\lrcorner$\else{\unskip\nobreak\hfil
\penalty50\hskip1em\null\nobreak\hfil$\lrcorner$
\parfillskip=0pt\finalhyphendemerits=0\endgraf}\fi} 

\newcommand{\cqed}{\renewcommand{\qed}{\cqedsymbol}}

\newcommand{\executeiffilenewer}[3]{%
\ifnum\pdfstrcmp{\pdffilemoddate{#1}}%
{\pdffilemoddate{#2}}>0%
{\immediate\write18{#3}}\fi%
} 

\newcommand{\tn}[1]{\tiny{#1}}

\newcommand{\Nats}{\ensuremath{\mathbb{N}}}

\newcommand{\Oh}{\ensuremath{\mathcal{O}}}

\newcommand{\tw}{\ensuremath{\mathtt{tw}}\xspace}
\newcommand{\gtw}{\ensuremath{\mathtt{gtw}}\xspace}
\newcommand{\pw}{\ensuremath{\mathtt{pw}}\xspace}

\newcommand{\guid}{\Lambda}

\newcommand{\bag}{\beta}
\newcommand{\sep}{\sigma}
\newcommand{\cone}{\gamma}
\newcommand{\comp}{\alpha}
\newcommand{\mrg}{\mu}

\newcommand{\dunion}{\uplus}

\newcommand{\atoms}{\mathbb{A}}

\newcommand{\hyptorso}{\mathsf{hypertorso}}

\newcommand{\glue}{\oplus}

\newcommand{\lintf}{\mathsf{left}}
\newcommand{\rintf}{\mathsf{right}}

\newcommand{\Gf}{\mathbb{G}}
\newcommand{\Hf}{\mathbb{H}}
\newcommand{\abst}[1]{\llbracket #1 \rrbracket}

\newcommand{\domain}{\varphi_{\mathsf{dom}}}
\newcommand{\inc}{\mathsf{incident}}

\newcommand{\Bag}{\mathsf{bag}}
\newcommand{\Parent}{\mathsf{parent}}
\newcommand{\decoration}{\mathsf{decoration}}

\newcommand{\Ss}{\mathcal S}
\newcommand{\Ii}{\mathcal I}

\newcommand{\Pp}{\mathcal P}

\newcommand{\Bb}{\mathcal B}
\newcommand{\Xx}{\mathcal X}
\newcommand{\mso}{{\sc mso}\xspace}
	
\newcommand{\set}[1]{\{#1\}}

\newcommand{\eqdef}{\stackrel {\text{def}} = }
\newcommand{\sucnodes}{\partial Z}

\setpagenumber{1}

\begin{document}
\title{Definability equals recognizability for graphs of bounded treewidth\thanks{
M. Boja\'nczyk is supported by the ERC Consolidator Grant LIPA. 
This work was partially done while Mi. Pilipczuk held a post-doc position at Warsaw Centre of Mathematics and Computer Science. Mi. Pilipczuk is also supported by the Foundation for Polish Science (FNP) 
via the START stipend programme.}}

\ifthenelse{\boolean{uglytwosided}}{
	\conferenceinfo{LICS '16}{July 5-8, 2016, New York City, USA}
	\copyrightyear{2016}
	\copyrightdata{978-1-nnnn-nnnn-n/yy/mm}
	\copyrightdoi{nnnnnnn.nnnnnnn}

	\authorinfo{  Miko\l{}aj Boja\'nczyk \and Micha\l{} Pilipczuk}{University of Warsaw}{\{bojan,michal.pilipczuk\}@mimuw.edu.pl}
}{
	\author{
  		Miko\l{}aj Boja\'nczyk
  		\thanks{
    			Institute of Informatics, University of Warsaw, Poland,
      			\texttt{bojan@mimuw.edu.pl}.
  		}
  		\and
  		Micha\l{} Pilipczuk
  		\thanks{
    			Institute of Informatics, University of Warsaw, Poland,
      			\texttt{michal.pilipczuk@mimuw.edu.pl}.
  		}
  	}
}

\maketitle

\begin{abstract}
We prove a conjecture of Courcelle, which states that a graph property is definable in \mso with modular counting predicates on graphs of constant treewidth if, and only if it is recognizable in the following sense: 
constant-width tree decompositions of graphs satisfying the property can be recognized by tree automata.
While the forward implication is a classic fact known as Courcelle's theorem, the converse direction remained open.
\end{abstract}

\ifthenelse{\boolean{uglytwosided}}{
\category{F.4}{Theory of Computation}{Mathematical Logic and Formal Languages}

\keywords
treewidth, tree decomposition, Monadic Second-Order Logic, recognizability, Simon's factorization forest
}{}

\section{Introduction}\label{sec:intro}
Classical results of B\"uchi, Trakhtenbrot and Elgot say that for finite words,  languages recognised by finite automata are exactly those definable in Monadic Second-Order logic \mso. Courcelle's theorem shows  the right-to-left inclusion holds for  graphs of bounded treewidth: if a property of graphs
 can be defined in \mso with quantification over edge subsets and
modular counting predicates --- henceforth called counting \mso on graphs --- then for every $k$ there is a tree automaton recognising tree decompositions of width $k$ of graphs satisfying the property.
A corollary is that model checking counting \mso on graphs of constant treewidth can be done in linear time, 
which is one of the foundational results in parameterized complexity. 
For more details, we refer the reader to the monograph of Courcelle and Engelfriet~\cite{CourcelleEngelfrietBook} devoted to \mso on graphs.

Courcelle's theorem, stated as above, generalizes only one direction of the equivalence between \mso and automata. The converse question is whether one can use counting \mso to define any property  of graphs that is recognizable for  constant treewidth, where recognizability is defined by, say, the finiteness of the index of an appropriate Myhill-Nerode relation. 
In the case of words (or trees), this is the `easy' direction: a formula of \mso can guess  an accepting run of an automaton, by  labelling nodes with states.
Surprisingly, the generalization of this implication to graphs of constant treewidth remained open for the last 25 years.

This problem, known as the {\em{Courcelle's conjecture}}, was initially formulated by Courcelle in the very first paper of his monumental series {\em{The monadic second-order logic of graphs}}~\cite{Courcelle90}.
The fifth paper of the series~\cite{Courcelle91a} is entirely devoted to its investigation and contains a proof for graphs of treewidth~2.
Since then, the conjecture has been confirmed for graphs of treewidth~$3$~\cite{Kaller00}, for $k$-connected graphs of treewidth~$k$~\cite{Kaller00}, 
for graphs of constant treewidth and chordality~\cite{BodlaenderHT15}, and for $k$-outerplanar graphs~\cite{JaffkeB15}.
There were also two claims of a significant progress on the general case.
First, Kabanets~\cite{Kabanets97} claimed a proof for graphs of bounded pathwidth, and then a resolution of the full conjecture was claimed by Lapoire~\cite{Lapoire98}.
Unfortunately, both these works are published only as extended conference abstracts, and no verified journal version has appeared.
The proofs of Kabanets and Lapoire are widely regarded as unsatisfactory; cf.~\cite{BodlaenderHT15,CourcelleEngelfrietBook,DowneyFellowsNewBook,JaffkeB15}.
In particular, the problem is stated as open both in the monograph of Courcelle and Engelfriet~\cite{CourcelleEngelfrietBook} and of Downey and Fellows~\cite{DowneyFellowsNewBook}.

The issue at the heart of Courcelle's conjecture  is that an \mso formula is applied to the graph alone, without access to any pre-defined tree decomposition.
Hence, one cannot simply guess a run of a tree automaton, because there is no tree.
As Courcelle puts it in~\cite{Courcelle91a}, {\emph{It is not clear at all how an automaton should traverse a graph. A ``general''
graph has no evident structure, whereas a word or a tree is (roughly speaking) its own algebraic structure}.
Hence, the natural approach to proving the conjecture is to find, using \mso and  the graph structure only, some tree decomposition of bounded width.
This strategy,  proposed by Courcelle~\cite{Courcelle91a}, was used in all the previous work on Courcelle's conjecture. We also use  this strategy.

Our main result is that there exists an {\em{\mso transduction}} that, given a graph of treewidth $k$ encoded as a relational structure, 
outputs an encoding of a tree decomposition of width $f(k)$, for some function~$f$.
Informally speaking, an \mso transduction guesses existentially a number of vertex and edge subsets, and based on them defines a tree decomposition.
Different guesses may lead to different decompositions, but provided the input graph has treewidth at most $k$, the output for at least one guess will be a tree decomposition of width bounded by $f(k)$.
The precise statement is in~Section~\ref{sec:main-res}.

\paragraph*{Acknowledgment. }We would like to thank Christoph Dittmann and Stephan Kreutzer for many helpful discussions.

\section{Overview}\label{sec:overview}
The main technical result of this paper, stated in Theorem~\ref{thm:compute-tree-decomposition},  is that using \mso one can compute a tree decomposition of a graph. This section describes the proof plan. We explain what it means to compute something in \mso, and  divide Theorem~\ref{thm:compute-tree-decomposition} into lemmas.

\subsection{Tree decompositions} 

In this section we define tree decompositions and show how we represent them for the purposes of \mso. 
Similar formalisms were used in the previous works, cf.~\cite{Courcelle91a,JaffkeB15,Kaller00}, but we choose to introduce our own language for the sake of being self-contained.

\paragraph*{Logical terminology.} Define a \emph{vocabulary} to be a set of relation names with associated arities (we do not use functions or constants). 
A \emph{logical structure} over a vocabulary consists of a universe supplied with interpretations of relations in the vocabulary.  We  use logical structures to model things like graphs and tree decompositions.
 
\paragraph*{Graphs as logical structures.} We model (undirected) graphs as logical structures, where the universe consists of both vertices and edges, 
and there is a binary \emph{incidence} relation $\inc(v,e)$ which says when a vertex $v$ is incident with an edge $e$.
The edges can be recovered as those elements of the universe that are second arguments of the incidence relation, and the vertices can be recovered as those elements of the universe that are not edges.
We do not allow multiple edges connecting the same pair of vertices, i.e., all graphs considered in this work are simple, unless explicitly stated.
We choose this  encoding so that set quantification in \mso can capture sets of edges as well. 


\paragraph*{Tree decompositions.} We begin by defining tree decompositions.
Define an {\em{in-forest}} to be an acyclic directed graph where every node has outdegree at most one. 
We use the usual tree terminology: root, parent, and child. 
Every connected component of an in-forest is a tree with exactly one root (vertex of outdegree zero).

\begin{definition}\label{def:tree-decomposition}
A \emph{tree decomposition} of a graph $G$ is an in-forest $t$ whose vertices called \emph{nodes}, and a labelling of the nodes  by sets of vertices in $G$, called \emph{bags}, subject to the following conditions:
\begin{itemize}
	\item for every edge $e$ of $G$, some bag contains both endpoints of $e$;
	\item for every vertex $v$ of $G$, the set of nodes in $t$ that are labelled by bags containing $v$ is nonempty and connected in $t$.
\end{itemize}
\end{definition}
Note two minor changes with respect to the classic definition: tree decompositions are rooted, and we allow them to be forests, instead of just trees.  
Both  changes are for convenience only and bear no significance for our results.

We write $x,y,z$ for nodes and $u,v,w$ for vertices. Below we introduce some terminology for tree decompositions, inspired by  Grohe and Marx~\cite{GroheM15}.

\begin{definition}\label{def:tree-decomposition-terminology}
	Let $x$ be a node in a tree decomposition $t$. The \emph{adhesion of $x$} is the intersection of the bags of $x$ and its parent; if $x$ is a root the adhesion is empty.  The \emph{margin of $x$} is its bag minus its adhesion.  The \emph{cone of $x$} is the union of the bags of the descendants of $x$, including $x$.  The \emph{component of $x$} is its cone minus its adhesion. If the decomposition $t$ is not clear from the context, we write $t$-cone, $t$-adhesion, etc.	
\end{definition}
Note that the margins of the nodes of a tree decomposition form a partition of the vertex set of the underlying graph.

A \emph{path decomposition} is the special case when the forest is a set of paths.
The \emph{width} of a tree or path decomposition is the maximum size of its bags, minus 1. 
The \emph{treewidth} of a graph is the minimum possible width of its tree decomposition, likewise for~\emph{pathwidth}. The treewidth and pathwidth of a graph $G$ are denoted by $\tw(G)$ and $\pw(G)$, respectively.

\paragraph*{Tree decompositions as logical structures.} 
A tree decomposition is represented as a logical structure as follows. 
The universe of the logical structure consists of the vertices and edges of the underlying graph, plus the nodes of the tree decomposition.
The vocabulary, which we call the \emph{vocabulary of tree decompositions}, consists of:
\begin{align*}
\underbrace{\mathsf{node}}_{\text{a unary predicate}} \quad \underbrace{\inc \quad \Bag \quad \Parent}_{\text{binary predicates}}
\end{align*}  
The predicate $\inc$ describes the incidence relation in the underlying graph. 
The predicate $\mathsf{node}(x)$ says that $x$ is a decomposition node,  $\Bag(v,x)$  says that vertex $v$ is in the bag of node $x$, 
and  $\Parent(x,y)$ says that node $x$ is the parent of node $y$. 

\paragraph*{\mso interpretations.}
\newcommand{\univformula}{\varphi_{\mathsf{univ}}}
We now define what it means to produce a tree decomposition (or some other structure) using \mso. 
We do this by using three types of basic operations on logical structures defined below, called \emph{copying}, \emph{coloring}, and \emph{interpreting}. 
All three types describe binary relations on logical structures: copying is a function, coloring is a relation, and interpreting is a partial function.
\begin{enumerate}
	\item {\it Copying.} Define the \emph{$k$-copy} of a logical  structure $\mathfrak A$ to be $k$ disjoint copies of $\mathfrak A$, with the following fresh predicates added to the vocabulary:
\begin{align*}
	\mathsf{copy}(a,b)\ ,\ \mathsf{layer}_1(a)\ ,\ \ldots\ ,\ \mathsf{layer}_k(a).
\end{align*}
The binary predicate $\mathsf{copy}$ checks whether two elements are copies of the same element of the original structure, whereas the unary predicate $\mathsf{layer}_i$ checks whether an element belongs to the $i$-th
copy (called also the $i$-th {\em{layer}}).
\item {\it Coloring.} 
Define an \emph{$k$-coloring} of a structure $\mathfrak A$ to be any structure obtained from $\mathfrak A$ by adding new unary predicates $X_1,\ldots,X_k$ to the vocabulary and 
interpreting them as any subsets of the universe.
\item {\it Interpreting.} The syntax of an interpretation consists of  
 an \emph{input vocabulary $\Sigma$}, an \emph{output vocabulary $\Gamma$} and a family of \mso formulas
	\begin{align*}
		\set{\domain, \univformula}\cup \set{\varphi_R}_{R\in \Gamma},
	\end{align*}
over the input vocabulary $\Sigma$. 
The formula $\domain$ has no free variables, the formula $\univformula$ has one free variable, and each formula $\varphi_R$ has as many free  variables as the arity of $R$. 
The free  variables in all of these formulas range over elements, not sets of elements. 
If $\mathfrak A$ is a logical structure over the input vocabulary $\Sigma$ that satisfies  $\domain$, we define the output logical structure, which is over the output vocabulary $\Gamma$, as follows. 
The universe is the universe of $\mathfrak A$ restricted to elements satisfying  $\univformula$ and each relation $R$ of the output vocabulary is interpreted as those tuples in the universe which make 
$\varphi_R$ true. If $\domain$ is not satisfied in $\mathfrak A$, then the output of the interpretation is undefined.
\end{enumerate}

\begin{definition}\label{def:nond-mso}
An \emph{ \mso transduction} is a finite composition of the three types of operation defined above (treated as relations between logical structures), 
together with prescribed input and output vocabularies.
If coloring is not used, we talk about a \emph{deterministic} \mso transduction.
If $\Ii$ is an  \mso transduction, and $\mathfrak A$ is a structure over the input vocabulary, then by $\Ii(\mathfrak A)$ we denote the {\em{output}} of $\Ii$, 
defined as the set of all structures over the output vocabulary that are in relation defined by $\Ii$ with $\mathfrak A$.
\end{definition}

The definition above is equivalent to the one in~\cite{0030804}; this equivalence follows from \cite[Theorem 1.39]{0030804}. 
The crucial property of  \mso transductions is the Backwards Translation Theorem~\cite[Theorem 1.40]{0030804}, which says that if $\Ii$  a  \mso transduction and  $\psi$ is an \mso sentence over the output vocabulary,  then 
\begin{align*}
	\set{\mathfrak A : \text{at least on structure in $\Ii(\mathfrak A)$ satisfies $\psi$}}
\end{align*}
is a set of structures over the input vocabulary that is definable in \mso (for completeness, we give a proof sketch adjusted to our notation as Lemma~\ref{lem:pull-back} in Appendix~\ref{app:overview}).
%
Using this result, we may apply \mso transductions
to enrich the input structure with \mso-definable objects, and any property that can be defined in \mso afterwards, can be also defined directly in the input structure.

\subsection{The main result}\label{sec:main-res}
We are now ready to  state our main technical result, which says that an \mso transduction can compute  tree decompositions for  graphs of bounded treewidth. 
We use the name \emph{transduction from graphs to tree decompositions} if the input vocabulary is the vocabulary of graphs $\{\inc(x,y)\}$ and the output vocabulary is the vocabulary of tree decompositions defined previously. 

\begin{theorem}\label{thm:compute-tree-decomposition} 
There is a function $f \colon \Nats \to \Nats$ such that for every $k \in \Nats$ there exists  an \mso transduction $\Ii$ from graphs to tree decompositions,
	 such that every graph $G$ satisfies:
	 \begin{enumerate}[(1)]
	 	\item\label{cnd:tw-output} Every output $\Ii(G)$ represents a tree decomposition of $G$ of width at most $f(k)$.
		\item\label{cnd:tw-nonempty} If  $G$ has treewidth at most $k$, then the output $\Ii(G)$ is nonempty.
	 \end{enumerate}
\end{theorem}

We actually believe that a stronger variant of the above theorem holds, with  $f$ being the identity. 
In other words, we believe that there is an \mso transduction which inputs a graph of treewidth $k$, and produces a tree decomposition of width $k$. 
In order to prove the stronger version, it would be sufficient to show that for every $k \le k'$, there is  an \mso transduction which realizes the following task: 
given a tree decomposition of width $k'$, produce, if possible, a tree decomposition of width $k$ for the same graph. 
Our idea for a proof of this statement is to take a closer look at the algorithm of Bodlaender and Kloks~\cite{BodlaenderK96} that solves exactly this task, 
and try to simulate it using  an \mso transduction (even a deterministic one).
We expand this topic in the concluding section (Section~\ref{sec:conclusions}).

The proof of Theorem~\ref{thm:compute-tree-decomposition} consists of two steps, described below.

\paragraph*{The special case of bounded pathwidth.}
The first step is to prove a weaker variant of the theorem. 
This variant has exactly the same statement, except that in condition~\eqref{cnd:tw-nonempty} the assumptions are strengthened to requiring that the pathwidth of the graph is at most~$k$. 
This weaker variant, Lemma~\ref{lem:compute-path-decomposition} below, is proved in Section~\ref{sec:pathwidth}. 

\begin{lemma}\label{lem:compute-path-decomposition}
There is a function $f \colon \Nats \to \Nats$ such that for every $k \in \Nats$ there exists  an \mso transduction $\Ii$ from graphs to tree decompositions,
such that every graph $G$ satisfies:
\begin{enumerate}[(1)]
\item\label{cnd:pw-output} Every output $\Ii(G)$ represents a tree decomposition of $G$ of width at most $f(k)$.
\item\label{cnd:pw-nonempty} If $G$ has pathwidth at most $k$, then the output $\Ii(G)$ is nonempty.
\end{enumerate}
\end{lemma}

There are two crucial ingredients in the proof of Lemma~\ref{lem:compute-path-decomposition}.

The first ingredient is that for path decompositions, we can use semigroup theory. Specifically, we use Factorisation Forest Theorem of Imre Simon~\cite{DBLP:journals/tcs/Simon90}. 
The application of this result is the cornerstone of our approach, and it enables us to recursively decompose any graph of pathwidth $k$ into constant-size pieces using only $f(k)$ levels of recursion --- 
a number that depends on $k$ alone, and not on the size of the graph. 
Lemma~\ref{lem:compute-path-decomposition} then follows by verifying that each level incurs a fixed blow-up of the width of tree decompositions that we are able to describe in~\mso.

The second ingredient is the definition of a \emph{guidance system}, which is a  combinatorial object used to  describe additional structure in a graph, e.g.,~a tree decomposition, 
in a way that can be guessed by a \mso transduction.
Guidance systems are introduced in Section~\ref{sec:guidance}, and are used throughout the whole paper to describe ``\mso-guessable'' tree decompositions.

\paragraph*{Tree decompositions with bags of bounded pathwidth.}
In the second step, presented in Section~\ref{sec:treewidth}, we show  that there is an \mso transduction which inputs a graph of treewidth at most $k$, 
and outputs a tree decomposition of the graph where the bags are maybe arbitrarily  large, but have pathwidth bounded by $2k+1$, in the following sense. 
For a node $x$ in a tree decomposition $t$, define its \emph{marginal graph} as follows: take the subgraph of the underlying graph induced by the margin of $x$, 
and add an edge $uv$ for every pair $\set{u,v}$ that appears together in the adhesion of some child node of $x$, provided this edge is not already included in the graph.

\begin{lemma}\label{lem:path-tree-decomposition}
For every $k \in \Nats$, there exists an \mso transduction $\Bb$ from graphs to tree decompositions
such that for every graph $G$ the following holds:
\begin{enumerate}[(1)]
\item\label{cnd:pw-tw-output} Every output $\Bb(G)$ represents a sane tree decomposition of $G$.
\item\label{cnd:pw-tw-nonempty} If $G$ has treewidth at most $k$, then $\Bb(G)$ contains at least one tree decomposition where all marginal graphs have pathwidth at most $2k+1$.
\end{enumerate}
\end{lemma}
In Lemma~\ref{lem:path-tree-decomposition} we use a technical notion of a {\em{sane tree decomposition}}, which is defined below.
\begin{definition}\label{def:sane}
A tree decomposition $t$ of a graph $G$ is called {\em{sane}} if the following conditions are satisfied for every node $x$:
\begin{enumerate}[(a)]
\item\label{p:notcontained} the margin of $x$ is nonempty;
\item\label{p:connectivity} the subgraphs induced in $G$ by the cone of $x$ and by the component of $x$ are connected;
\item\label{p:neighbors} every vertex of the adhesion of $x$ has a neighbor in the component of $x$.
\end{enumerate}
\end{definition}
Intuitively, saneness means that the decomposition  respects the connectivity of subgraphs corresponding to its subtrees.
Indeed, it is straightforwards to see from the definition, that all the marginal graphs of a sane decomposition are nonempty and connected (connectivity follows from property~\eqref{p:connectivity} of saneness).
A similar notion of {\em{internal connectivity}} was used by Lapoire~\cite{Lapoire98}. 
The following lemma, which may be considered folklore, shows  that any tree decomposition can be  sanitized.
\begin{aplemma}\label{lem:sanitation}
Suppose $t$ is a tree decomposition of a graph $G$. Then there exists a sane tree decomposition $s$ of $G$ where every bag in $s$ is a subset of some bag in $t$.
In particular, if $\tw(G)\leq k$, then $G$ admits a sane tree decomposition of width at most $k$.
\end{aplemma}
The above lemma is colored red because its proof can be found in the appendix. Throughout the paper, we use this convention to mark statements whose proofs are straightforward or not important for the main narrative;
their formal verification is deferred to the appendix in order not to spoil the natural flow of argumentation.

%


\paragraph*{Proof of Theorem~\ref{thm:compute-tree-decomposition}.} The proof of Theorem~\ref{thm:compute-tree-decomposition} is a combination of 
Lemmas~\ref{lem:compute-path-decomposition} and~\ref{lem:path-tree-decomposition}. This requires some technical care, but does not involve any substantially new ideas. We give a full exposition of this proof 
in Appendix~\ref{sec:wrapup}.

\subsection{Courcelle's conjecture}
\label{sec:courcelle}
In this section we use Theorem~\ref{thm:compute-tree-decomposition} to prove the conjecture of Courcelle mentioned in the introduction.
We use a syntax slightly different than Courcelle.

\begin{definition}\label{def:interface-graph}
Define an \emph{interface graph} to consist of an \emph{arity} $k \in \Nats$, an \emph{underlying graph} $G$ called the \emph{underlying graph} and  an \emph{interface mapping}, which is an injective partial function from $\set{1,2,\ldots,k}$ to vertices of the underlying graph. If  image of  $i\in \set{1,2,\ldots,k}$ under the interface mapping is defined, it is called the \emph{$i$-th interface vertex}. Then $i$ is the {\em{name}} of this interface vertex.
\end{definition}
	
Interface graphs of arity $k$ are called $k$-interface graphs. If $\Gf$ and $\Hf$ are $k$-interface graphs, then their \emph{gluing} $	\Gf \glue \Hf$
 is the $k$-interface graph defined as follows. 
 The underlying graph is the disjoint union of the two underlying graphs, with the interface vertices having the same names in $\Gf$ and $\Hf$ being fused.
 In other words, for any name $i\in \set{1,2,\ldots,k}$ that is used both in $\Gf$ and in $\Hf$ (i.e. is in the intersection of the domains of the interface mappings), 
 we fuse the corresponding $i$-th interface vertices in $\Gf \glue \Hf$. 
 If this process creates any parallel edges, we remove the duplicates.
 The names of the interface vertices in $\Gf \glue \Hf$ are inherited from the arguments. 
 We illustrate this definition with the following example:

\begin{center}
	\includegraphics[scale=0.42]{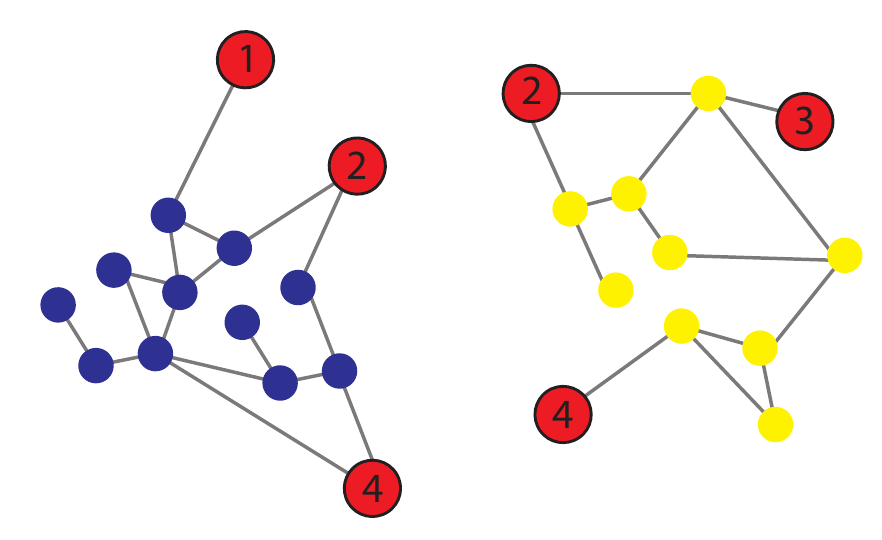} \quad \raisebox{1cm}{{\Large{$\leadsto$}}} \quad	\includegraphics[scale=0.42]{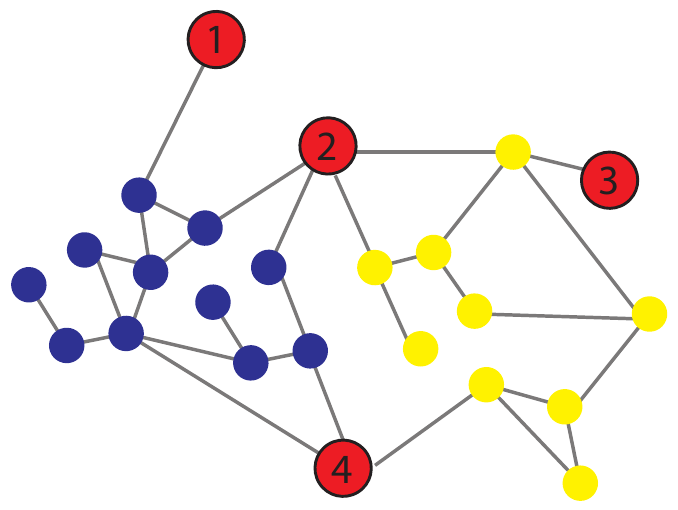}
\end{center}

In Courcelle's syntax from~\cite{Courcelle91a}, the gluing essentially corresponds to substituting $\Hf$ for the hyperedge consisting of the interface vertices inside $\Gf$.

\paragraph*{Recognisability.} Let $\Pi$ be a property of graphs.
We say that two $k$-interface graphs $\Gf_1$ and $\Gf_2$ are $\Pi$-equivalent if 
\begin{align*}
	\Gf_1 \oplus \Hf \mbox{ satisfies } \Pi \qquad\mbox{iff} \qquad \Gf_2 \oplus \Hf\mbox{ satisfies $\Pi$}.
\end{align*}
holds for every $k$-interface graph $\Hf$.
This is an  equivalence relation on $k$-interface graphs.
If there are finitely many equivalence classes, then we say that  $\Pi$ is \emph{$k$-recognisable}.
Finally, we say that $\Pi$ is  \emph{recognisable} if it is $k$-recognisable for every  $k$. This is equivalent to Definition 1.12 in~\cite{Courcelle91a}. 

\paragraph*{Courcelle's conjecture.} In what follows, we consider the logic {\em{counting \mso}} which is the extension of \mso on graphs by predicates of 
the form ``the size of set $X$ is divisible by $m$'' for every constant~$m$. The following result was stated as Conjecture 1 in~\cite{Courcelle91a}.

\begin{theorem}\label{thm:courcelle-conjecture}
If a property of graphs that have treewidth bounded by a constant is recognisable, then it is definable in counting \mso.
\end{theorem}

As mentioned in the introduction, the converse implication was  proved by Courcelle in~\cite{Courcelle90}. 
In his later work~\cite{Courcelle91a}, Courcelle proposed the following approach to proving Theorem~\ref{thm:courcelle-conjecture}. 
Call a property of graphs $\Pi$  {\em{strongly context-free}} if (informally), given any graph $G$ from $\Pi$, some constant-width decomposition of $G$ can be nondeterministically defined in \mso.
If we prove that the class of graphs of treewidth $k$ is strongly context-free (which is stated as Conjecture 2 in~\cite{Courcelle90}),
then Theorem~\ref{thm:courcelle-conjecture} would follow from the following lemma.

\begin{aplemma}\label{lem:courcelle}
Let $k \in \Nats$ and let $\Pi$ be a $k$-recognisable property of graphs. 
There is a formula of counting \mso over the vocabulary of tree decompositions which is true in exactly those structures which represent a tree decomposition of width $k$ where the underlying graph satisfies $\Pi$.
\end{aplemma}

Lemma~\ref{lem:courcelle} is essentially proved in~\cite{Courcelle91a} (cf.~Theorem~4.8 therein), but for the sake of completeness we give a proof adjusted to our notation in the appendix.
The statement that the class of graphs of treewidth at most $k$ is strongly context-free is, up to insignificant differences in definitions, equivalent to our Theorem~\ref{thm:compute-tree-decomposition}.
Hence, we can complete the proof of Theorem~\ref{thm:courcelle-conjecture} as Courcelle suggested.

\begin{proof}[Proof of Theorem~\ref{thm:courcelle-conjecture}]
Let $\Pi$ be a property of graphs of treewidth at most $k$. 
Apply Theorem~\ref{thm:compute-tree-decomposition} to $k$, yielding an \mso transduction, which maps graphs of treewidth at most $k$ to tree decompositions of width at most $f(k)$. 
Apply Lemma~\ref{lem:courcelle} to $f(k)$ and the property $\Pi$, yielding a formula of counting \mso which tests $\Pi$ on tree decompositions of width at most $f(k)$. 
The result follows by using the Backwards Translation Theorem.
\end{proof}

We believe that a stronger statement holds, namely: if $\Pi$ is a property of graphs of treewidth $k$, then already being $k$-recognisable implies definability in counting \mso. 
This claim would follow from the stronger version of  Theorem~\ref{thm:compute-tree-decomposition} described after its statement.



\section{Guidance systems}\label{sec:guidance}
In this section, we introduce guidance systems. The definition is a variant of the guidance systems defined in~\cite{BojanczykL12}. 
Guidance systems are used in the proofs of  Lemmas~\ref{lem:compute-path-decomposition} and~\ref{lem:path-tree-decomposition}, which are found in Sections~\ref{sec:pathwidth} and~\ref{sec:treewidth}.

\begin{definition}\label{def:guidance-system}
A {\em{guidance system}} $\guid$ over a graph $G$ is a family of in-trees (i.e.~connected in-forests), where each in-tree is 
obtained by orienting the edges of some subgraph of $G$. For a vertex $u$, define
\begin{align*}
	\Lambda(u) = \set{v : \mbox{some tree from $\Lambda$ contains $u$ and has root $v$}}.
\end{align*}
A \emph{coloring} of a guidance system is an assignment of trees to colors so that trees with the same color have disjoint vertex sets.
\end{definition}

\begin{example}\label{example:guidance-first-example}
Let $G$ be an undirected cycle of length six, with vertices called $\set{0,1,\ldots,5}$, and edges being neighbor modulo 6.  
Consider the following guidance system 
\begin{align*}
\Lambda =	\set{u \to u+1 \to u+2  : u \in \set{0,\ldots,5}}
\end{align*}
where addition is modulo 6. This guidance system is 3-colorable:
\begin{center}
	\includegraphics[scale=0.5]{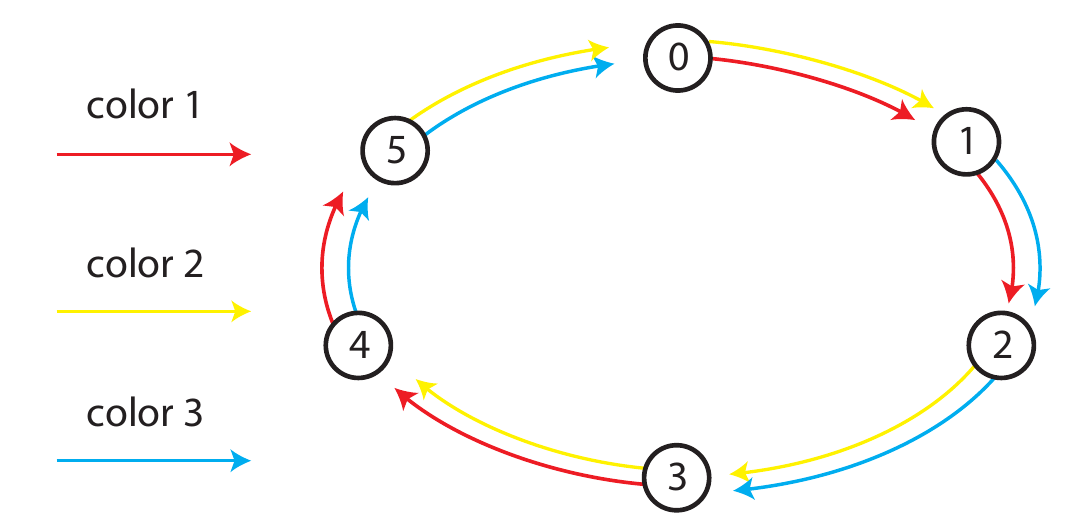}
\end{center}		
For this guidance system, 
\begin{align*}
	\Lambda(u) = \set{u,u+1,u+2 \mod 6} \qquad \mbox{for }u \in \set{0,\ldots,6}.
\end{align*}
\end{example}

Guidance system are used to recognize sets of vertices in a graph, in the sense  defined below.
\begin{definition}\label{def:capturing}
Let $G$ be a graph. 
A set of vertices $X$ is said to be \emph{captured} by a vertex $u$ in a guidance system $\Lambda$ if $X \subseteq \Lambda(u)$ holds. 
A family of sets of vertices is said to be captured by $\Lambda$ if each set is captured by some vertex.
\end{definition}

%
%
%

\paragraph*{From adhesions to a tree decomposition.}
In Lemma~\ref{lem:capture-adhesion-capture-tree-decomposition} below, we show   that  in order to produce a  sane tree decomposition with an \mso interpretation, it suffices to capture all its adhesions with a bounded number of colors. 
Note that in the lemma below  we do not restrict the sizes of bags in the tree decompositions; e.g.,~a decomposition with all vertices in one bag has no adhesions and therefore falls into the scope of the lemma.

\begin{aplemma}\label{lem:capture-adhesion-capture-tree-decomposition}
For every $k \in \Nats$ there is an \mso transduction from graphs to tree decompositions which maps every graph  $G$ to all sane tree decompositions of $G$ whose 
family of adhesions can be captured by a $k$-colorable guidance system over~$G$.
\end{aplemma}

The proof of Lemma~\ref{lem:capture-adhesion-capture-tree-decomposition}, which can be found in Appendix~\ref{app:guidance-long-proof}, is actually quite non-trivial.
We  use the connectivity conditions given by saneness in order to be able to guess the bags of a sane tree decomposition.
Also, instead of the original graph, we need to work with the graph obtained by turning all adhesions into cliques. 
This structure can be constructed by  an \mso transduction which guesses a guidance system that captures all the adhesions.
The fact that the guidance system can be colored using few colors is necessary for it to be guessable in \mso.


%
%
%
%

\paragraph*{Stability under small modifications.}  In Lemma~\ref{lem:hyper-pilipczuk} below, we show that the number of colors needed to capture a family of sets by a guidance system is stable under removing or adding vertices. 
If $G$ is a graph, we write $G-u$ for the subgraph induced by removing $u$ from the vertices. 

\begin{aplemma}\label{lem:hyper-pilipczuk}
	Let $G$ be a graph, let $\Xx$ be a family of subsets of vertices, and let $u$ be a vertex.
	\begin{enumerate}[(1)]
			\item\label{cnd:hyp-pil:left} If every set in $\Xx$ is contained in some connected component of $G$ and  
			      \begin{align*}
					\Xx - u \eqdef	\set{ X - \{u\} : X \in \Xx}
			      \end{align*}  
			      is captured by a $k$-colorable guidance system over $G$, then $\Xx$ is captured by a $(k+1)$-colorable guidance system over $G$.
			\item\label{cnd:hyp-pil:right} If every set from $\Xx$ is contained in some connected component of $G-u$ and $\Xx$ is captured by a $k$-colorable guidance system over $G$, 
			      then $\Xx$ is captured by a $2k$-colorable guidance system over $G-u$.
	\end{enumerate}
\end{aplemma}

\section{Graphs of bounded pathwidth}\label{sec:pathwidth}
In this section we prove Lemma~\ref{lem:compute-path-decomposition}, which says that  an \mso transduction can transform a graph of bounded pathwidth into a tree decomposition. 
Our proof relies on the guidance systems defined in the previous section. We first outline the plan. 
In Section~\ref{sec:guided-treewidth}, we define a graph parameter called guided treewidth. 
In Section~\ref{sec:guided-treewidth-bounded-by-pathwidth} we show that bounded pathwidth implies bounded guided treewidth. 
A combination of this result with the fact that guidance system can be expressed in \mso (Lemma~\ref{lem:capture-adhesion-capture-tree-decomposition}) yields Lemma~\ref{lem:compute-path-decomposition}.

\subsection{Guided treewidth}
\label{sec:guided-treewidth}

The following definition can be seen as a new graph parameter.
\begin{definition}\label{def:guided-tree-width} 
Define the \emph{guided treewidth} of a graph $G$, denoted by $\gtw(G)$, 
to be the smallest $k$ such that there exists a tree decomposition of $G$ where all bags are captured by some $k$-colorable guidance system over $G$.
\end{definition}
Note that if a bag is captured by a $k$-colorable guidance system, then it has size at most $k$. 
Hence, the guided treewidth of a graph is an upper bound on its treewidth.
Since adhesions are contained in bags, tree decompositions whose bags are captured by a $k$-colorable guidance system fall into the scope of 
Lemma~\ref{lem:capture-adhesion-capture-tree-decomposition}, and can be produced by an  \mso transduction.
%
%
%
%
%
%

The goal of this section is to show that bounded  pathwidth implies bounded guided treewidth, and therefore, by the above discussion, 
tree decompositions of graphs of bounded pathwidth can be produced by an \mso transduction.

To illustrate guided treewidth, we show an example where a poorly chosen tree decomposition needs a large number of colors to be captured.

\begin{example}\label{example:guided-cycles}
Consider a cycle with vertices $\set{1,\ldots,2n}$. 
For this cycle, consider a path decomposition with $n-1$ bags, where the $i$-th bag contains vertices $i$, $i+1$, $2n-i$, and $2n-i+1$.
A picture for $n=5$ is in the left panel of the figure below:
\begin{center}
	\includegraphics[scale=0.4]{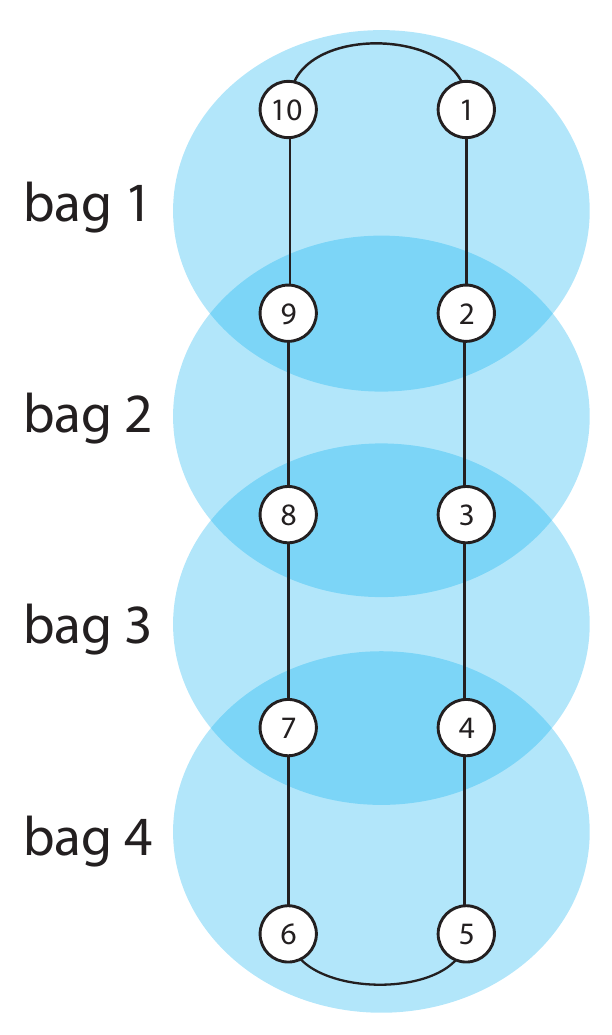} \qquad \includegraphics[scale=0.4]{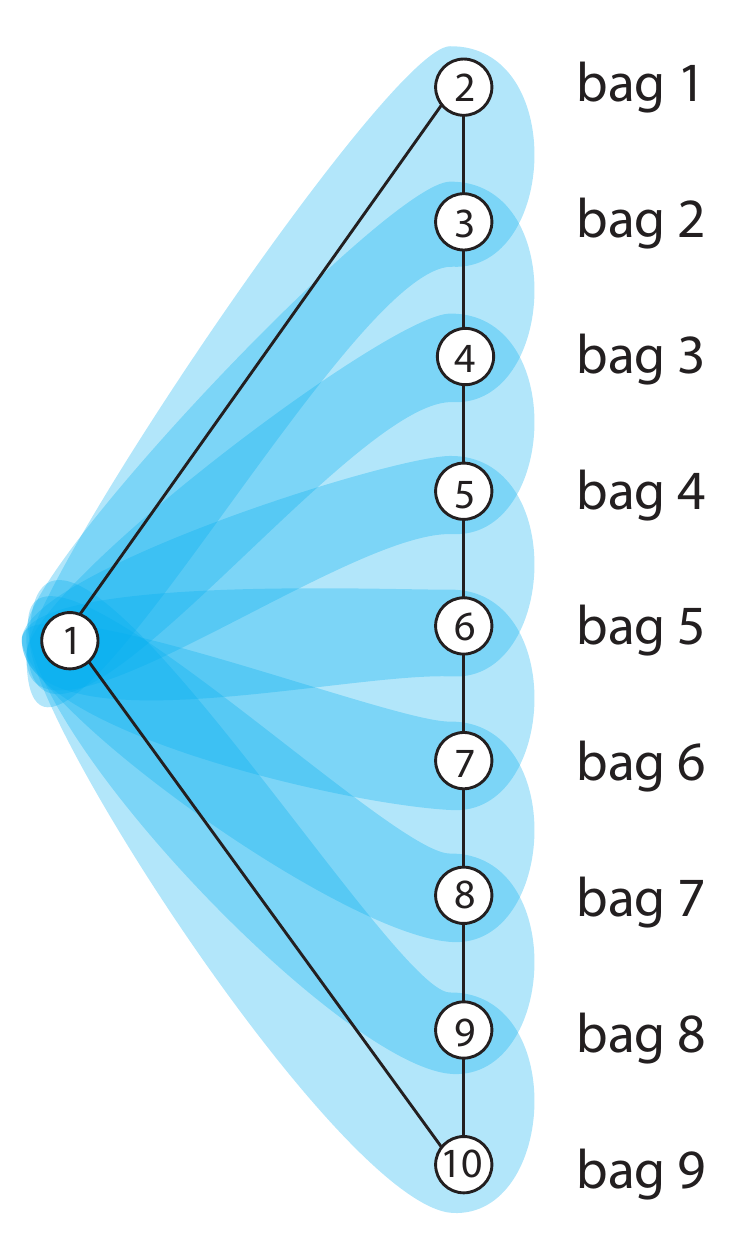}
\end{center}
The path decomposition above has suboptimal width, but we could make it optimal by extending the graph by a disjoint clique of size~4.

It is now not hard to verify that any guidance system $\Lambda$ that captures this tree decomposition, requires a number of colors that is linear in~$n$.
Indeed, for $i\in \set{1,\ldots,n-1}$, let $u_i$ be the vertex that captures the bag $\{i,i+1,2n-i,2n-i+1\}$ in $\Lambda$.
Hence, there are four in-trees, respectively with roots $i$, $i+1$, $2n-i$, and $2n-i+1$, such that each of them contains $u_i$.
Consequently, wherever $u_i$ lies, it must hold that one of these trees contains one of the following four arcs: $(v_1,v_{2n})$, $(v_{2n},v_1)$, $(v_{n},v_{n+1})$, or $(v_{n+1},v_n)$.
If we now restrict our attention only to odd indices $i$, the corresponding bags of the decomposition are disjoint, and hence one of arcs above belongs to at least $\frac{n-1}{8}$ different in-trees of $\Lambda$.
These trees must receive pairwise different colors in any coloring of $\Lambda$.

To fix the problem we use a different path decomposition, with $2n-2$ bags, such that the $u$-th bag contains $\set{1,u,u+1}$. 
This decomposition is depicted in the right panel of the figure above.
To capture its bags, we use a 3-colorable guidance system colorable. The first color is used to describe a directed path 
\begin{align*}
	2 \to 3 \to \cdots \to 2n \to 1.
\end{align*}
The remaining two colors are used alternately to connect each vertex with its successor, similarly as in Example~\ref{example:guidance-first-example}. 
Concluding, each cycle admits a tree (even path) decomposition that can be captured by a 3-colorable guidance system.
\end{example}

In the next section, we  strengthen the result from the above example, and show that bounded pathwidth implies bounded guided treewidth. 
Before passing to the next section, we show that guided treewidth is robust with respect to graph operations like disjoint union (denoted $\dunion$) or adding/removing a single vertex.
The proof is a simple application of Lemma~\ref{lem:hyper-pilipczuk}.
\begin{aplemma}\label{lem:basic-properties-of-gtw}
	Let $G,G'$ be graphs and let $u$ be a vertex in $G$. Then
	\begin{eqnarray}
		\gtw(G\dunion G') &=& \max(\gtw(G),\gtw(G'))\label{eq:basic-dunion}\\ 
		\gtw(G) &\leq& \gtw(G-u)+1\label{eq:basic-left-rem}\\ 
		\gtw(G-u) &\leq& 2\cdot \gtw(G) \label{eq:basic-right-rem}
	\end{eqnarray}
\end{aplemma}

\subsection{Guided treewidth is bounded by pathwidth}
\label{sec:guided-treewidth-bounded-by-pathwidth}
We now state and prove the main result of Section~\ref{sec:pathwidth}, which is that bounded pathwidth implies bounded guided treewidth.
\begin{lemma}\label{lem:pathwidth-comb}
There exists a function $f(k)\in 2^{2^{\Oh(k^2)}}$ such that 
\begin{align*}
	\gtw(G)\leq f(\pw(G)) \qquad \mbox{for every graph $G$}.
\end{align*} 
\end{lemma}

Note the asymmetry in the lemma: we assume bounded pathwidth, but produce a tree decomposition. 
It can be easily seen that Lemma~\ref{lem:compute-path-decomposition} follows by combining Lemma~\ref{lem:pathwidth-comb} with Lemmas~\ref{lem:sanitation} and~\ref{lem:capture-adhesion-capture-tree-decomposition}.
A full proof of this implication is in Appendix~\ref{app:pathwidth}.

The rest of Section~\ref{sec:pathwidth} is devoted to proving Lemma~\ref{lem:pathwidth-comb}. 
In Section~\ref{sec:interface-graphs} we define bi-interface graphs, which give an alternative algebraic definition of pathwidth. 
Then, in Section~\ref{sec:simon}, we use the Factorization Forest Theorem to prove Lemma~\ref{lem:pathwidth-comb}.

\subsubsection{Bi-interface graphs}
\label{sec:interface-graphs}
Recall the interface  graphs as defined in Section~\ref{sec:courcelle}. 
In our approach to pathwidth, we use such an enriched version of this definition, where a graph is supplied with two sets of interfaces: left and right.
Here is the formal definition.  

\begin{definition}\label{def:twosided-interface}
A \emph{bi-interface graph} consists of an \emph{arity} $k \in \Nats$,  an \emph{underlying graph $G$},  and  two partial injective functions $\lintf,\rintf$ from $\set{1,\ldots,k}$ to the vertices of $G$.
We use the name \emph{$i$-th left interface} for  $\lintf(i)$, likewise for the $i$-th right interface.
Moreover, we require that if a vertex is simultaneously an $i$-th left interface and a $j$-th right interface, then $i=j$.
\end{definition}

Thus, we assume that the interface names of a bi-interface graph of arity $k$ are numbers between $1$ and $k$, however not all of them need to be used.

In  Section~\ref{sec:courcelle} we showed how to glue interface graphs. For bi-interface graphs we use a similar notion, which is probably best seen in a picture, see Figure~\ref{fig:gluing}. 
\begin{figure*}[htbp]
	\centering
	\begin{tabular}{ccc}
	\includegraphics[scale=0.5]{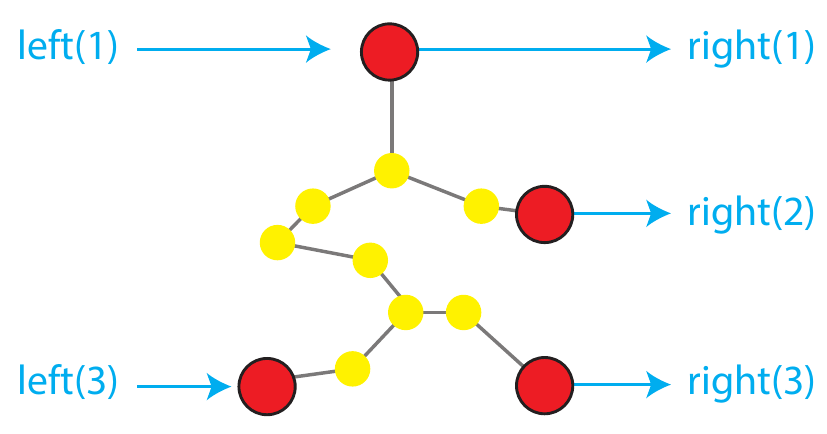} \qquad& 
		\includegraphics[scale=0.5]{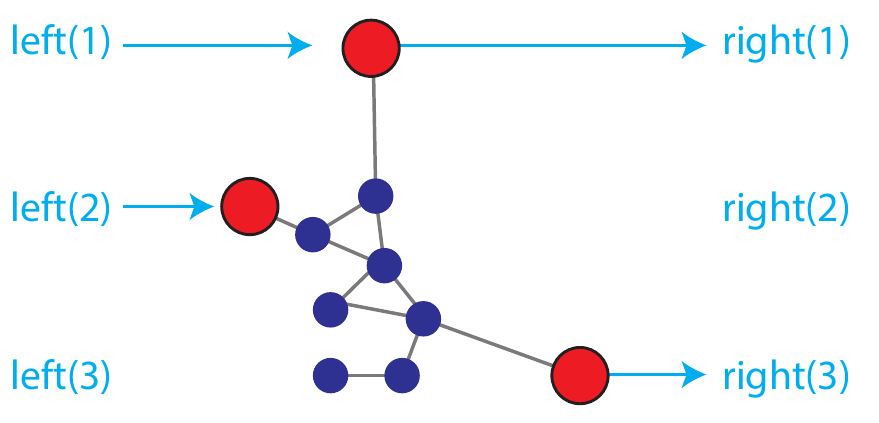}\qquad &
		\includegraphics[scale=0.5]{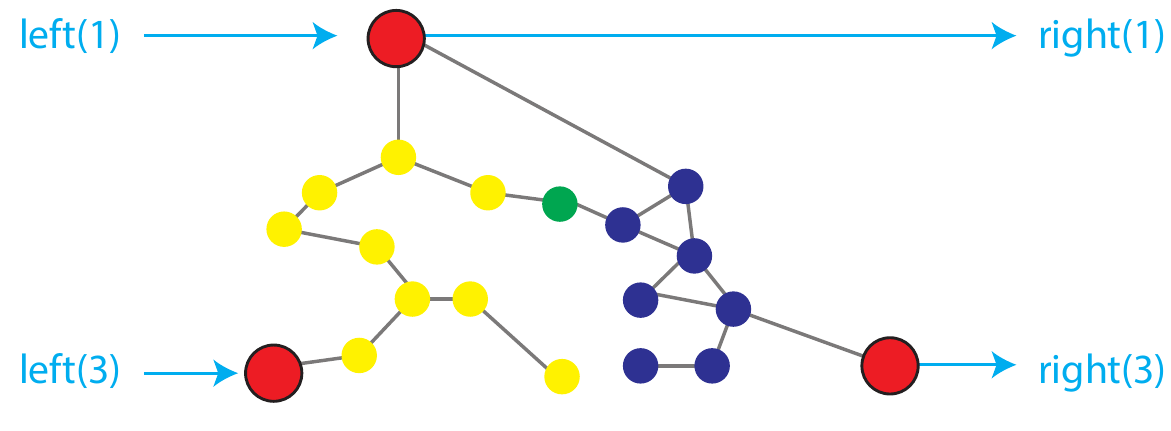} \\
		$\Gf_1$ & $\Gf_2$ &	$\Gf_1 \glue \Gf_2$
	\end{tabular}
	
	\caption{Two bi-interface graphs and their gluing. 
		 The interface nodes are the red ones, the incoming arrows indicate left interfaces and the outgoing arrows indicate right interfaces. 
		 Note how some of the left or right interfaces are undefined, e.g.~the second left interface in $\Gf_1$, and how the first left and right interfaces are equal.		 %
 }
	\label{fig:gluing}
\end{figure*}
Here is the formal definition. Let $\Gf_1,\Gf_2$ be two bi-interface graphs of the same arity $k$.  
Define their gluing $\Gf_1 \glue \Gf_2$ as follows.
Take the disjoint union of the underlying graphs, and fuse the $i$-th right interface of $\Gf_1$ with the $i$-th left interface of $\Gf_2$, whenever both are defined.
As before, remove the duplicates whenever any parallel edge is created in this operation.
As the left interface function take the left interface function of $\Gf_1$, and as the right interface function take the right interface function of $\Gf_2$.
It is easy to verify that if $\Gf_1$ and $\Gf_2$ were both bi-interface graphs of arity $k$, then $\Gf_1 \glue \Gf_2$ is also a bi-interface graph of arity $k$.
Note that  in $\Gf_1 \glue \Gf_2$ we forget the information about the right interfaces of $\Gf_1$ and the left interfaces of $\Gf_2$.

The gluing operation  defined above is associative, turning the  set of bi-interface graphs of arity $k$ into a semigroup. A  product 
\begin{align*}
	\Gf_1\glue  \ldots\glue \Gf_n
\end{align*}
in this semigroup
is essentially the same thing as a path decomposition, where the bags are the bi-interface graphs $\Gf_1,\ldots,\Gf_n$, and the interface functions say how the bags are connected. 
Hence the following lemma, whose straightforward proof is omitted.

\begin{lemma}\label{lem:pw-semigroup}
A graph has pathwidth at most $k$ if, and only if it is the underlying graph of a bi-interface graph of the form
\begin{align*}
	\Gf_1\glue \ldots\glue \Gf_n
\end{align*}
where each $\Gf_i$ has arity $k$ and at most $k+1$ vertices.
\end{lemma}

\paragraph*{Abstraction.} Call two bi-interface graphs \emph{isomorphic} if there is a bijection between their vertex sets that respects graph edges and the name of each interface.
For a graph $G$ and a subset of vertices $X$, by the {\em{torso}} of $G$ with respect to $X$ we mean a graph on vertex set $X$ where two vertices are adjacent if they can be connected in $G$ by a path
whose internal vertices do not belong to $X$.
Define the \emph{abstraction} $\abst \Gf$ of a bi-interface graph $\Gf$ to be the isomorphism type of the bi-interface graph obtained by taking the torso of the underlying graph with 
respect to the interface vertices, and preserving the interface functions. Here is a picture of a bi-interface graph and its abstraction:
\begin{center}
	\begin{tabular}{cc}
	\includegraphics[scale=0.4]{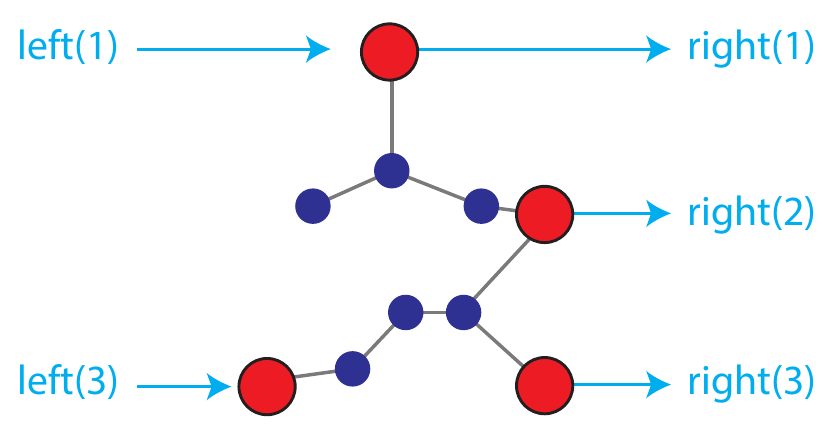} & 
		\includegraphics[scale=0.4]{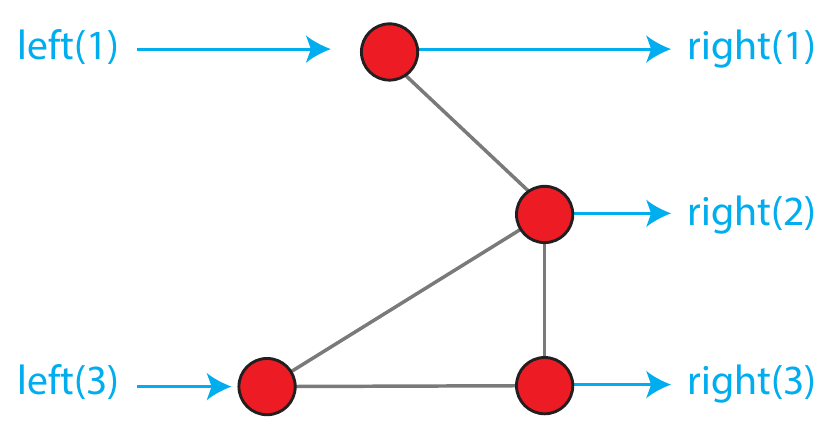} \\
			$\Gf$ & $\abst \Gf$
	\end{tabular}
\end{center}
Define $\atoms_k$ to be the possible abstractions of bi-interface graphs of arity $k$. This is a finite set of size $2^{\Oh(k^2)}$.
The abstraction function $\Gf \mapsto \abst \Gf$ is compositional in the sense that $\abst{\Gf_1 \glue \Gf_2}$ is uniquely determined by $\abst{\Gf_1}$ and $\abst{\Gf_2}$. 
This means that $\atoms_k$ can be uniquely endowed with a multiplication operation which makes the abstraction function a semigroup homomorphism from bi-interface graphs of arity $k$ to $\atoms_k$. 
From now on we will treat $\atoms_k$ as a semigroup, and the abstraction function as a homomorphism.

\subsubsection{Simon's Factorization Forest Theorem}
\label{sec:simon}

We now recall Simon's Factorization Forest Theorem from semigroup theory, whose application is the cornerstone of our proof of Lemma~\ref{lem:compute-path-decomposition}. 
Intuitively, the theorem says that if $h$ is a homomorphism from words into a finite semigroup $S$, 
then every word can be recursively factorised until reaching single letters, with the depth of the factorisation depending only on the size of $S$, 
and each factorisation step respecting $h$ in some way.

We write $\Sigma^+$ for the semigroup of nonempty finite words over  alphabet $\Sigma$ with concatenation.
Let $S$ be a finite semigroup and let 
\begin{align*}
	h : \Sigma^+ \to S
\end{align*} 
be a  semigroup homomorphism. 
We define two types of factorizations for a word $u\in \Sigma^+$.
\begin{itemize}
\item {\bf{Binary}}: a factorization into two  factors
\begin{align*}
	u = u_1u_2.
\end{align*}
\item {\bf{Unranked}}:  a factorization into an arbitrary number of  factors
\begin{align*}
	u = u_1 \ldots u_n,
\end{align*}  such that all  factors $u_1,\ldots,u_n$ have the same image under $h$.
\end{itemize}
The {\em{$h$-rank}} of a word $u\in \Sigma^+$ is the natural number defined by induction on the length of $u$ as follows. Single letters have rank 1. For a longer word $u$, its $h$-rank is
\begin{align*}
1+\min_{u=u_1\ldots u_n}\, \max_{i \in \set{1,\ldots,n}}\, \text{$h$-rank of $u_i$}
\end{align*}
where the minimum ranges over factorizations (either binary or unranked) of $u$ into nonempty factors. 
Using binary factorizations only, it is easy to see that every word has $h$-rank that is at most logarithmic in its length.  Imre
Simon proved that thanks to the unranked factorisations, one can achieve constant rank.

\begin{theorem}[Factorisation Forest Theorem~\cite{DBLP:journals/tcs/Simon90, DBLP:conf/mfcs/Kufleitner08}]\label{thm:Simon}
If $h \colon \Sigma^+ \to S$ is a semigroup homomorphism with $S$ finite, then every word in $\Sigma^+$ has $h$-rank at most $3|S|$.
\end{theorem}

The original result there is actually slightly stronger: in the unranked factorisations only idempotent values in the semigroup can be used. 
We will not use idempotence.
The bound $3|S|$ is from~\cite{DBLP:conf/mfcs/Kufleitner08}, improving on the original exponential bound of Simon~\cite{DBLP:journals/tcs/Simon90}. 
The word \emph{forest} in the theorem's name is because the definition of rank implicitly uses a tree structure of~factorisations.


We now return to the proof of Lemma~\ref{lem:pathwidth-comb}, which says that guided treewidth is bounded by a function of the pathwidth. Consider  the semigroup homomorphism
\begin{align*}
	h : \Sigma^+ \to \atoms_k,
\end{align*}
where $\Sigma$ is the set of arity-$k$ bi-interface graphs with at most $k+1$ vertices, and $h$ is the homomorphism that glues a word to a single bi-interface graph and takes its abstraction. 
We will show that for~every
\begin{align*}
	\Gf_1 \cdots \Gf_n \in \Sigma^+
\end{align*}
the guided treewidth of the corresponding glued graph is at most exponential in the $h$-rank of the word. More precisely:
\begin{align}\label{eq:induction-on-rank}
	\gtw(\Gf_1 \glue \cdots \glue \Gf_n) \le 2^{\Oh({k \cdot \text{($h$-rank($\Gf_1 \cdots \Gf_n$)}}))}.
\end{align}
By Lemma~\ref{lem:pw-semigroup} and the Factorisation Forest Theorem, every graph of pathwidth $k$ can be obtained from some word in $\Sigma^+$ with $h$-rank bounded by $3|\atoms_k|=2^{\Oh(k^2)}$, 
thus proving Lemma~\ref{lem:pathwidth-comb}. 
It remains to show~\eqref{eq:induction-on-rank}. 
For this, we use induction on the $h$-rank. 
The induction base follows from the observation that the guided treewidth of a graph is at most the number of its vertices, which follows directly from claim~\eqref{eq:basic-left-rem} of Lemma~\ref{lem:basic-properties-of-gtw}.
For the induction step, we use the Binary and Unranked lemmas stated below. The Binary lemma, used for binary factorizations, follows immediately from Lemma~\ref{lem:basic-properties-of-gtw}.

\begin{lemma}[Binary]\label{lem:binary}
If $\Gf_1,\Gf_2$ are arity-$k$ bi-interface graphs,~then
\begin{align*}
      \gtw(\Gf_1\glue \Gf_2)\leq k+2^k\cdot \max(\gtw(\Gf_1),\gtw(\Gf_2)).
\end{align*}
\end{lemma}
\begin{proof}
Observe that $\Gf_1\glue \Gf_2$ can be obtained from $\Gf_1$ and $\Gf_2$ by (i) removing from both graphs those interface vertices that get fused during gluing, (ii) taking the disjoint union of the 
obtained graphs, and (iii) reintroducing the removed interface vertices to the result. 
By claim~\eqref{eq:basic-right-rem} of Lemma~\ref{lem:basic-properties-of-gtw}, operation (i) can increase the guided treewidth of each of the graphs by a multiplicative factor of at most $2^k$. 
By claim~\eqref{eq:basic-dunion} of Lemma~\ref{lem:basic-properties-of-gtw}, in operation (ii) we obtain a graph with guided treewidth not exceeding the maximum of the guided treewidth of the factors. 
Finally, by claim~\eqref{eq:basic-left-rem} of Lemma~\ref{lem:basic-properties-of-gtw}, in operation (iii) the guided treewidth increases by an additive factor of at most $k$. Hence the bound follows.
\end{proof}

For unranked factorisations, we use the following lemma.
\begin{lemma}[Unranked]\label{lem:unranked}
	Let $\Gf_1,\Gf_2,\ldots,\Gf_n$ be bi-interface graphs of arity $k$ which all have the same abstraction. Then
	$$\gtw(\Gf_1 \glue \cdots \glue \Gf_n)\leq k(4k^2+5)+8^k\cdot \max_{i\in \set{1,\ldots,n}}\, \gtw(\Gf_i).$$
\end{lemma}
The proof of the Unranked Lemma is more involved and we present it in the following section.

\subsubsection{Proof of the Unranked Lemma}
\label{sec:unranked}

\newcommand{\Lintf}{I^{\mathrm{left}}}
\newcommand{\Rintf}{I^{\mathrm{right}}}

\newcommand{\colfun}{\mathsf{column}}
We begin by introducing some notation.
If  $\Gf_1, \ldots, \Gf_n$ are bi-interface graphs of the same arity and $i \in \set{1,\ldots,n}$, then there is a natural injective mapping $\colfun_i$ that associates the vertices of $\Gf_i$ 
with the corresponding vertices of its copy in $\Gf_1 \glue \cdots \glue \Gf_n$. 
Let us call this function the \emph{$i$-th column function}, and its image the \emph{$i$-th column}. Here is a picture that illustrates columns:
\begin{center}
	\includegraphics[scale=0.45]{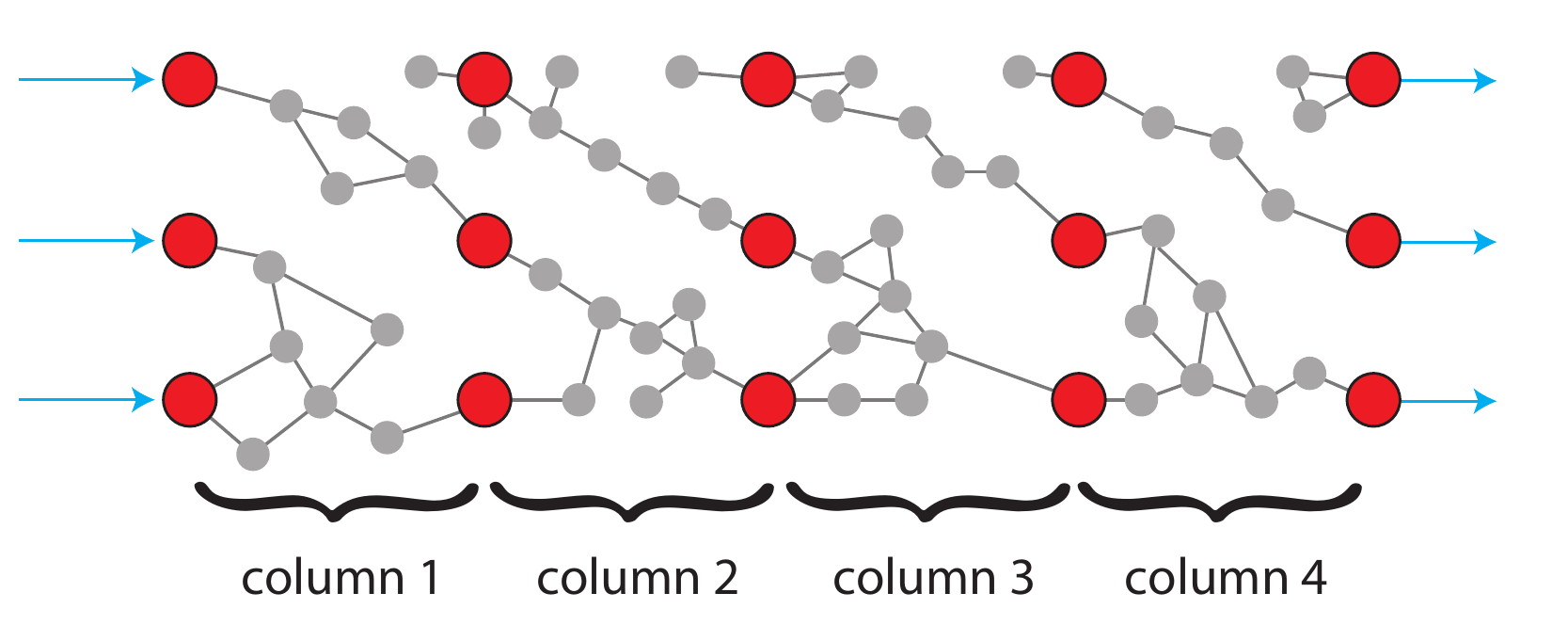}
\end{center}

We begin by observing that if we glue many copies of the same bi-interface graph, then vertices in the same column can either be connected by a path visiting few columns, or not at all.
Here, by a path {\em{staying within}} a set of columns we mean that the vertex set of the path is contained in the union of these columns.
\begin{lemma}\label{lem:basic-locality}
	Let $\Gf$ be a bi-interface graph. If two vertices of 
	\begin{align*}
	\Gf^n \eqdef \overbrace{\Gf \glue \cdots \glue \Gf}^{\text{$n$ times}}.
	\end{align*}
	are in the same column  $i \in \set{1,\ldots,n}$ and can be connected by a path,  
	then they  can also be connected by a path that stays only in columns $j$ satisfying $|j-i| \le k^2$, where $k$ is the arity of $\Gf$.
\end{lemma}
\begin{proof}
Consider a path in $\Gf^n$ which connects two vertices from the $i$-th column. 
Cut this path whenever it visits an interface vertex of the $i$-th column, thus decomposing it into subpaths of three different types:
	\begin{enumerate}[(i)]
		\item\label{type:stay} stays within column $i$;
		\item\label{type:left} connects a pair of left interfaces of column $i$, and stays within columns $<i$;
		\item\label{type:right} connects a pair of right interfaces of column $i$, and stays within columns $>i$.
	\end{enumerate}
We show that a path of each type can be modified so that it stays only within columns with indices differing by at most $k^2$ from $i$. 
For type~\eqref{type:stay}, this is already true.
By symmetry, we only treat type~\eqref{type:left}. We assume that $i>k^2+1$, because otherwise we are already done.

For $m\in \Nats$, define 
	\begin{align*}
		\mathrm{Left}_m \subseteq \set{1,\ldots,k} \times \set{1,\ldots,k}
	\end{align*}
to be the set of pairs $(x,y)$ such that there is a path in $\Gf^n$ which connects the $x$-th left interface in column $m$ with the $y$-th left interface in column $m$, and stays within columns $<m$.
Observe that set $\mathrm{Left}_{m+1}$ is defined only in terms of $\mathrm{Left}_m$; formally, there is a function $\Phi$ that depends on $\Gf$ only such that $\mathrm{Left}_{m+1}=\Phi(\mathrm{Left}_{m})$.
This is because any path certifying that some pair belongs to $\mathrm{Left}_{m+1}$ can be decomposed into subpaths staying within column $m$ and connecting some its interfaces,
and subpaths contributing to the definition of $\mathrm{Left}_{m}$, i.e., certifying that some pair belongs to $\mathrm{Left}_{m}$.
Hence, whether some pair belongs to $\mathrm{Left}_{m+1}$ depends only on the existence of connections between interfaces of $\Gf$ within this graph, and the information encoded in $\mathrm{Left}_{m}$.
Moreover, we have that $\mathrm{Left}_{m+1}\supseteq \mathrm{Left}_{m}$, because any path certifying that some pair belongs to $\mathrm{Left}_{m}$ can be shifted one column to the right 
(due to all columns being isomorphic), and then it certifies that the same pair belongs to $\mathrm{Left}_{m+1}$. Therefore we have that
\begin{align*}
	\mathrm{Left}_1 \subseteq \mathrm{Left}_2 \subseteq \mathrm{Left}_3 \subseteq  \ldots
\end{align*}
Since each of these sets has size at most $k^2$, the sequence must contain two equal consecutive sets after at most $k^2+1$ steps. 
Moreover, since $\mathrm{Left}_{m+1}$ is defined only in terms of $\mathrm{Left}_{m}$, the sequence must stabilize from this point on.
We conclude that $\mathrm{Left}_j=\mathrm{Left}_{k^2+1}$ for all $j>k^2$, and in particular $\mathrm{Left}_i=\mathrm{Left}_{k^2+1}$.
This means that any path of type~\eqref{type:left}, say connecting the $x$-th and $y$-th interface of column $i$, 
can be replaced by taking a path certifying that pair $(x,y)$ belongs also to $\mathrm{Left}_{k^2+1}$, and shifting it
$i-(k^2+1)$ columns to the right. Then the replaced path stays within the $k^2$ columns to the right of column $i$.
\end{proof}

\newcommand{\guidglue}{\guid_{\glue}}
We now show that if we glue many bi-interface graphs which are not necessarily equal but have the same abstraction, 
then there is a guidance system, colorable by a small number of colors, which takes each vertex to the reachable interfaces from its own column.

\begin{lemma}\label{lem:general-locality}
Let $\Gf_1,\Gf_2,\ldots,\Gf_n$ be bi-interface graphs with the same arity $k$ and the same abstraction, and such that the left and right interfaces are disjoint. 
Consider the gluing
\begin{align*}
	\Gf = \Gf_1 \glue \cdots \glue  \Gf_n.
\end{align*}
For $i \in \set{1,\ldots,n}$ and a vertex $u$ in the $i$-th column, define $I_i(u)$ to be those interface vertices of the $i$-th column which are reachable from $u$ by a path in  $\Gf$. 
Then there is a guidance system $\Lambda$ over $\Gf$, which can be colored by at most $4k(k^2+1)$ colors, such that
\begin{align*}
	I_i(u) \subseteq \Lambda(u) \quad \text{for each $i \in \set{1,\ldots,n}$ and $u$ in the $i$-th column.}
\end{align*}
\end{lemma}
\begin{proof}
Using Lemma~\ref{lem:basic-locality} applied to the common abstraction of graphs $\Gf_i$, and the definition of the abstraction operation, we get the following claim.
\begin{claim}\label{cl:local} For every $v \in I_i(u)$ there is a path from $u$ to $v$ which stays only within columns $j$ satisfying $|i-j| \le k^2$.\end{claim}
Let  $i \in \set{1,\ldots,n}$ and  let $v$ be an interface vertex in the $i$-th column. 
Consider the subgraph of $\Gf$ induced by all vertices that can be reached from $v$ via paths that stay within columns $j$ with $|i-j| \le k^2$; this subgraph is connected by definition.
Choose an arbitrary spanning tree of this subgraph, and call it $t_{i,v}$. 
Choose $v$ to be the root of this tree, and direct all edges of the tree toward the root. 
Define $\Lambda$ to be the family
\begin{align*}
	\set{{t_{i,v}} : {\text{$i \in \set{1,\ldots,n}$ and $v$ is an interface in the $i$-th column}}}
\end{align*}
By Claim~\ref{cl:local}, we know that  $v \in I_i(u)$ implies that $u$ belongs to the tree $t_{i,v}$, and therefore $I_i(u) \subseteq \Lambda(u)$.
By the assumption on the left and right interfaces being disjoint, columns of $\Gf$ are disjoint unless they are consecutive.
Therefore, the trees $t_{i,v}$ and $t_{j,w}$ are disjoint whenever $|i-j|$ is greater than $2k^2$. 
Finally, in each column there are at most $2k$ interface vertices. 
This means that $\Lambda$ can be colored using  $(2k^2+1) \cdot 2k$ colors as follows: each tree $t_{i,v}$ receives a color depending on the name of interface $v$ (left or right, and a number between $1$ and $k$)
and the remainder of $i$ modulo $2k^2+1$.
\end{proof}

The next step is to prove the Unranked Lemma in the special case where the bi-interface graphs have disjoint left and right interfaces. The proof is an application of Lemma~\ref{lem:general-locality}, which
ensures us that connections that have to be realized by the sought guidance system have a local character.

\begin{lemma}\label{lem:unranked-almost}
	Let $\Gf_1,\Gf_2,\ldots,\Gf_n$ be bi-interface graphs of arity $k$ which all have the same abstraction, and such that the left and right interfaces are disjoint. Then
	$$\gtw(\Gf_1 \glue \cdots \glue \Gf_n)\leq 4k(k^2+1)+4^k\cdot \max_{i\in \set{1,\ldots,n}}\, \gtw(\Gf_i).$$
\end{lemma}
\begin{proof}
By claim~\eqref{eq:basic-dunion} of Lemma~\ref{lem:basic-properties-of-gtw}, it suffices to show that the guided treewidth is small for every connected component of
\begin{align*}
	\Gf \eqdef \Gf_1 \glue \cdots \glue \Gf_n.
\end{align*}
Let $X$ be the vertex set of a connected component of $\Gf$. 
For a column $i \in \set{1,\ldots,n}$ in $\Gf$, define  $X_i$ to be the intersection of $X$ with the $i$-th column, and define $Y_i \subseteq X_i$ to be the interfaces of the $i$-th column that are inside $X$. 
Iteratively apply claim~\eqref{eq:basic-right-rem} of Lemma~\ref{lem:basic-properties-of-gtw} to the subgraph induced by $X_i$, removing from it the at most $2k$ vertices of $Y_i$.
This yields a tree decomposition, call it $t_i$, of the graph induced by $X_i-Y_i$, such that $t_i$ is captured by a guidance system, call it $\Lambda_i$, that uses at most 
\begin{align*}
	4^k \cdot \gtw(\Gf_i)
\end{align*}
colors. 
Define $s_i$ to be the tree decomposition obtained from $t_i$ by adding the interface vertices $Y_i$ to every bag. 
Define a tree decomposition $t$ as follows.
First, for each index $i$, create a node $y_i$ with $Y_i$ being the associated bag, and attach all the roots of decomposition $s_i$ as children of the node $y_i$. 
Then connect nodes $y_i$ into a path: for each index $i<n$, attach $y_{i+1}$ as a child of $y_i$. 
It is not hard to see that $t$ is a tree  decomposition of the connected component of $\Gf$ induced by $X$.
	 
We now define a guidance system that captures all bags of $t$. 
Apply Lemma~\ref{lem:general-locality} to $\Gf_1,\ldots,\Gf_n$  yielding a guidance system; call it $\guidglue$. 
Add this guidance system to the guidance systems $\Lambda_1,\ldots,\Lambda_n$ defined above, defining a guidance system
\begin{align*}
\Lambda \eqdef \guidglue \cup \Lambda_1 \cup \cdots \cup \Lambda_n.	
\end{align*} 
Since columns can only intersect on interfaces, it follows that the guidance systems $\Lambda_1,\ldots,\Lambda_n$ are disjoint. 
Therefore $\Lambda_1  \cup \cdots \Lambda_n$ can be colored by the maximum number of colors needed for each $\Lambda_1,\ldots,\Lambda_n$. 
To this, we add the number of colors needed for $\guidglue$ from Lemma~\ref{lem:general-locality}, thus getting the bound in the statement of the lemma. 
To finish the proof, we show that $\Lambda$ captures every bag in the tree decomposition $t$. 
Let us first consider the bags of the form~$Y_i$. By Lemma~\ref{lem:general-locality}, we know that 
\begin{align}\label{eq:from-general-locality}
	Y_i \subseteq \guidglue(u) \qquad \mbox{for every $u \in X_i$}.
\end{align}
In particular, $Y_i$ is captured by any vertex from itself. Consider now a bag $B$ in one of the tree decompositions $s_i$.  By assumption on $\Lambda_i$, there is some vertex $u \in X_i$ such that 
\begin{align*}
	B - Y_i \subseteq \Lambda_i(u).
\end{align*}
Combining the above with~\eqref{eq:from-general-locality}, we get that $B$ is captured by $u$ in the guidance system $\guid$.
\end{proof}

We now finish the proof of the Unranked Lemma. Let 
\begin{align*}
	\Gf \eqdef \Gf_1 \glue \cdots \glue \Gf_n
\end{align*}
be a gluing of bi-interface graphs of the same arity $k$ and the same abstraction.
Define $U_i$ to be the vertices in $\Gf_i$ that are both left and right interfaces at the same time and define $\Hf_i$ to be $\Gf_i$ with $U_i$ removed. 
Consider the gluing
\begin{align*}
	\Hf \eqdef \Hf_1 \glue \cdots \glue \Hf_n.
\end{align*}
By the assumption that the abstractions of all of $\Gf_1,\ldots,\Gf_n$ are the same, it follows that the image of $U_i$ under the $i$-th column function is the same set, independent of $i$, 
and this set has at most $k$ vertices. 
Therefore, $\Hf$ is equal to $\Gf$ with at most $k$ vertices removed. By claim~\eqref{eq:basic-left-rem} of Lemma~\ref{lem:hyper-pilipczuk} we infer that
\begin{align}
\gtw(\Gf)\leq \gtw(\Hf)+k.\label{eq:Gf-Hf}
\end{align}
The left and right interfaces are disjoint in each $\Hf_i$, and these graphs also have the same abstraction, so we can use Lemma~\ref{lem:unranked-almost} to~get:
\begin{align}
\gtw(\Hf)\leq 4k(k^2+1)+4^k\cdot \max_{i\in \set{1,\ldots,n}}\, \gtw(\Hf_i).\label{eq:Hf-Hfi}
\end{align}
Finally, each $\Hf_i$ is obtained from $\Gf_i$ by removing at most $k$ vertices, and hence we can apply claim~\eqref{eq:basic-right-rem} of Lemma~\ref{lem:hyper-pilipczuk} to infer that
\begin{align}
\gtw(\Hf_i)\leq 2^k\cdot \gtw(\Gf_i).\label{eq:Hfi-Gfi}
\end{align}
By combining \eqref{eq:Gf-Hf},~\eqref{eq:Hf-Hfi}, and~\eqref{eq:Hfi-Gfi}, we obtain the desired upper bound on the guided treewidth of $\Gf$.
This finishes the proof of the Unranked Lemma and, as argued before, also the proof of Lemma~\ref{lem:pathwidth-comb} is thus complete.

\section{Graphs of bounded treewidth}\label{sec:treewidth}
\newcommand{\Inc}{I}
\newcommand{\Cut}{\mathsf{cut}}
\newcommand{\Pfam}{\mathcal{P}}

\newcommand{\dcmp}{\mathtt{dcmp}}
\newcommand{\Req}{\mathcal{R}}
\newcommand{\tb}{t^{\circ}}

In this section we prove Lemma~\ref{lem:path-tree-decomposition}. Our strategy is to show the following result on guidance systems.

\begin{lemma}\label{lem:low-pw-decomp}
Let $G$ be a graph of treewidth at most $k$. Then $G$ admits a tree decomposition $\tb$ with the following properties:
\begin{enumerate}[(a)]
\item\label{p:pwbound-gl} every marginal graph  has pathwidth at most~$2k+1$;
\item\label{p:cngbound-gl} the family of adhesions of $\tb$ can be captured by a guidance system colorable with $4k^3+2k$ colors.
\end{enumerate}
\end{lemma}

From the lemma above, we can deduce Lemma~\ref{lem:path-tree-decomposition} as follows. 
Take interpretation $\Ss_{4k^3+2k}$ from Lemma~\ref{lem:capture-adhesion-capture-tree-decomposition}. 
 Lemma~\ref{lem:low-pw-decomp} implies that if the treewidth of the input graph is at most $k$, then at least one output of $\Ss_{4k^3+2k}$ satisfies the conditions of~Lemma~\ref{lem:path-tree-decomposition}.

The rest of Section~\ref{sec:treewidth} is devoted to proving Lemma~\ref{lem:low-pw-decomp}. 
In Section~\ref{sec:local-path-tree}, we prove a ``local'' variant of the lemma, which provides one step of the construction. 
In Section~\ref{sec:global}, we iterate the local variant to construct a global decomposition, thus proving~Lemma~\ref{lem:low-pw-decomp}.

\subsection{A local variant of Lemma~\ref{lem:path-tree-decomposition}}
\label{sec:local-path-tree}
In this section, we state and prove Lemma~\ref{lem:local-decomp}, which can be seen as a local variant of Lemma~\ref{lem:low-pw-decomp}. 
We begin with some hypergraph terminology, which is used in its statement and proof.

\paragraph*{Hypergraphs.}
A \emph{hypergraph} consists of a set of vertices together with a multiset of nonempty subsets of the vertices, called the \emph{hyperedges}. 
Note the use of   multisets: hyperedges can appear multiple times.
If $H$ is a hypergraph, the hypergraph \emph{induced} by a subset of vertices  $X$, denoted by $H[X]$, is the hypergraph where 
the vertices are $X$  and the hyperedges are intersections of the original hyperedges  with $X$,  with the empty ones removed. 
A \emph{path} in a hypergraph is a sequence $$(u_1,e_1,u_2,\ldots,u_p,e_p,u_{p+1}),$$ where $u_i$ are pairwise different vertices, $e_i$ are pairwise different edges, 
and vertices $u_i,u_{i+1}$ are contained in hyperedge $e_i$ for each $i=1,2,\ldots,p$.
Vertices $u_1$ and $u_{p+1}$ are the {\em{endpoints}}, and the path is said to \emph{go} from $u_1$ to $u_{p+1}$.
Each vertex $u_i$ and hyperedge $e_i$ is said to be {\em{traversed}} by the path.
Connected components in a hypergraph are defined by path-connectedness in a natural manner. 
Tree decompositions of hypergraphs are defined as for graphs, except that every hyperedge must be contained in some bag.


\paragraph*{Prefixes and their torsos.} Let $t$ be a sane tree decomposition of a graph $G$. 
A {\em{prefix}} of $t$ is a set of nodes $Z$  in $t$ that is closed under taking ancestors.  If $Z$ is a prefix, define $\partial Z$ to be the nodes of $t$ that are not in $Z$, but their parent is in $Z$. For a prefix $Z$, define
\begin{align*}
	\hyptorso(t,Z)
\end{align*}
to be the  hypergraph obtained by taking the subgraph of $G$ induced by the union of the bags in $Z$, and then for every $z \in \partial Z$ a   hyperedge for the adhesion of  $z$.
When adding hyperedges, we respect multiplicities, i.e.~if $n$ nodes in $\partial Z$ have the same adhesion, then this adhesion is used $n$ times in $\hyptorso(t,Z)$.

We are now ready to state the local version of Lemma~\ref{lem:low-pw-decomp}.
\begin{lemma}\label{lem:local-decomp}
Let $t$ be a width $k$ sane decomposition of a connected graph $G$. Let $u,v$ be  vertices in the root bag (note that there is a unique root due to connectedness).
Then there exists a nonempty prefix $Z$ of $t$ with the following properties:
\begin{enumerate}[(a)]
\item\label{p:pwbound} the pathwidth of $\hyptorso(t,Z)$ is at most $2k+1$, and
\item\label{p:twopaths} the vertices $u$ and $v$ can be connected by two paths in $\hyptorso(t,Z)$ such that if a hyperedge is traversed by both paths, then it is an edge of $G$.
\end{enumerate}
\end{lemma}

The rest of Section~\ref{sec:local-path-tree} is devoted to proving the above lemma. The proof uses a Menger style argument, so we begin by defining  terminology for hypergraph networks.
\paragraph*{Networks.} Define a {\em{network}} to be a connected hypergraph together with two distinguished different vertices, called the  {\em{source}} and the {\em{sink}}. 
We extend all notation from hypergraphs to networks in a natural manner.
A \emph{cutedge} in a network is a hyperedge which appears in every path from the source to the sink; equivalently, the removal of a cutedge makes the source and the sink fall into different connected components.
The following claim is straightforward.
\begin{aplemma}\label{lem:order}
In every network one can order all of the cutedges into a sequence $(e_1,\ldots,e_p)$
such that every path from the source to the sink visits the cutedges in this order.
\end{aplemma}

The sequence $(e_1,\ldots,e_p)$ yielded by Lemma~\ref{lem:order} will be called the {\em{cutedge sequence}} of a network.
Define a \emph{cutedge component} in a network to be a connected component of the hypergraph obtained by removing the cutedges.
The following claim is straightforward.
\begin{aplemma}\label{lem:incidence-cutedges}
Consider a network. Let $(e_1,\ldots,e_p)$ be its cutedge sequence, and define $e_0,e_{p+1}$ be the singletons of the source and the sink, respectively.
Then every cutedge component intersects exactly one or exactly two among elements $e_0,e_1,\ldots,e_p,e_{p+1}$. 
In the former case, the intersected $e_i$ is a cutedge; i.e. $i$ is not equal to $0$ or $p+1$.
In the latter case, the two intersected elements must be consecutive in the sequence.
\end{aplemma}
The above lemma motivates the following terminology for a cutedge component. 
If it intersects two consecutive elements $e_i$ and $e_{i+1}$ in the cutedge sequence extended by the singletons of the source and the sink, then it is called an \emph{$(e_i,e_{i+1})$-bridge}. 
If a cutedge component is not a bridge of any kind, and therefore it intersects exactly one cutedge $e_i$, then it is called an \emph{$e_i$-appendix}. 
This terminology is illustrated in Figure~\ref{fig:decomp}.
 \begin{figure}[htbp!]
         \centering
                \includegraphics[width=\columnwidth]{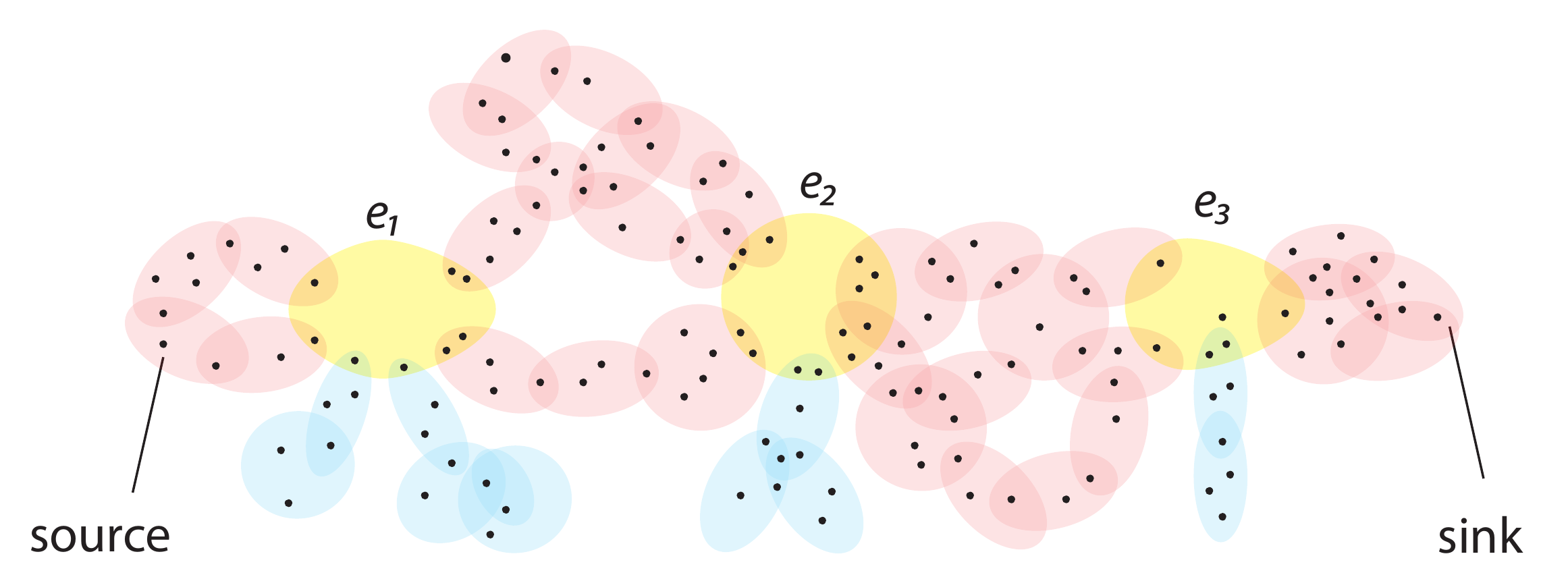}
 \caption{A network. The yellow hyperedges are the cutedges. The red hyperedges are those that contain bridges, the blue ones are those that contain appendices. 
 Note the vertex in $e_3$ which participates in no blue or red hyperedges, this vertex is a singleton $e_3$-appendix.}\label{fig:decomp}
 \end{figure}
 
%

%
The following lemma is proved by applying  Menger's theorem, for every $i \in \set{0,1,\ldots,p}$, to the union of all $(e_i,e_{i+1})$-bridges. 

\begin{aplemma}\label{lem:menger-str} 
In every network one can find two paths from the source to the sink, such that a hyperedge is traversed by both paths if, and only if it is a cutedge.
\end{aplemma}

\paragraph*{The invariant.} To prove Lemma~\ref{lem:local-decomp}, we will start with a prefix of the decomposition that contains only the root, and keep on extending the prefix until it satisfies condition (b). 
While extending the prefix, we will preserve an invariant  which  implies condition (a). 
We now describe the invariant. 
Let $t,u,v$ be as in Lemma~\ref{lem:local-decomp}.
For a prefix $Z$ of $t$, consider the network obtained from $\hyptorso(t,Z)$ by choosing $u$ as the source and $v$ as the sink. 
We will maintain the invariant that this network is $k$-thin in the sense defined below.


\begin{definition}\label{def:thin}
Consider a network with hypergraph $H$ and cutedge sequence $(e_1,\ldots,e_p)$; we follow the convention that $e_0,e_{p+1}$ denote the singletons of the source and the sink, respectively.
Define $V_i$ to be the union of the vertex sets of all $(e_i,e_{i+1})$-bridges, for $i=0,1,\ldots,p$, and $W_i$ to be the union of the vertex sets of all $e_i$-appendices, for $i=1,\ldots,p$. The network is called $k$-thin if:
\begin{enumerate}[(a)]
\item for every $i \in \set{0,1,\ldots,p}$, the induced hypergraph $H[V_i]$ admits a path decomposition of width at most $2k+1$
where the first bag contains $e_i\cap V_i$ and the last bag contains $e_{i+1}\cap V_i$;
\item for every $i \in \set{1,2,\ldots,p}$, the induced hypergraph $H[W_i]$ admits a path decomposition of width at most $k$ where the first bag contains $e_i\cap W_i$.
\end{enumerate}
\end{definition}

The following lemma shows that our invariant implies condition~\eqref{p:pwbound} in Lemma~\ref{lem:local-decomp}. 
The proof is a simple surgery on decompositions certifying thinness.

\begin{aplemma}\label{lem:thin-pw}
A $k$-thin network has pathwidth at most $2k+1$.
\end{aplemma}

Before finishing the proof of Lemma~\ref{lem:local-decomp}, we show that thinness is preserved under a certain kind of replacements. 
Let $H,K$ be hypergraphs such that the intersection of their vertex sets is equal to a hyperedge $e$ of $H$. 
Define  $H[e\to K]$ to be the following hypergraph. 
The vertex set is the union of the vertex sets in $H,K$. 
The hyperedges are the multiset union of the hyperedges in $H,K$, with the hyperedge $e$ removed. 
The following lemma shows that the above defined replacement preserves $k$-thinness of networks, assuming that $e$ is a cutedge and $K$ is small. 
The proof  is a technical, though conceptually simple analysis of the relationship between cutedges before and after the replacement. 

\begin{aplemma}\label{lem:thin-persist}
Consider a $k$-thin network with hypergraph $H$. Let $K$ be a connected hypergraph with at most $k+1$ vertices, 
with no hyperedge larger than $k$, and such that the intersection of the vertices of $H$ and $K$ is equal to some cutedge $e$ of $H$. 
Then the hypergraph $H[e\to K]$, with the same source and sink as in $H$, is also a $k$-thin~network.
\end{aplemma}

We are now ready to prove Lemma~\ref{lem:local-decomp}.
\begin{proof}
For a prefix $Z$ of the tree decomposition $t$, define the \emph{network of $Z$} to be the network obtained from $\hyptorso(t,Z)$ by choosing the source to be $u$ and the sink to be $v$. 
This is indeed a network: the underlying hypergraph is connected because $G$ itself is connected. 

Initially, choose $Z$ to be the prefix that contains only the root node of $t$.  
We will maintain the invariant that the network of $Z$ is $k$-thin. 
The invariant is clearly satisfied by the initial choice, because the root bag has size at most $k+1$, and adhesions have sizes at most $k$ (due to the saneness of $t$).

By Lemma~\ref{lem:thin-pw}, the invariant implies condition~\eqref{p:pwbound}. 
We show below that if $Z$ is a prefix satisfying the invariant, then either it satisfies condition~\eqref{p:twopaths}, 
in which case we are done, or one can add a node to the prefix while maintaining the invariant. 
Since the tree decomposition $t$ has a finite number of nodes, this process  has to stop at some moment, thus proving the lemma.

Let then $Z$ be a prefix such that the network of $Z$ is $k$-thin. 
Apply Lemma~\ref{lem:menger-str}, yielding two paths from the source to the sink in the network of $Z$, such that the only hyperedges traversed by both paths are the cutedges of the network of $Z$.
If all these cutedges are original edges of $G$, then we are done, because $Z$ satisfies condition~\eqref{p:twopaths}. 
Otherwise, there is some cutedge $e$ in the network of $Z$ that is not an edge of $G$. 
By definition of the network of $Z$, the cutedge $e$ corresponds to the adhesion of some $z \in \partial Z$.
Again by definition, the network of $Z \cup \set z$ is obtained from the network of $Z$ by: 
adding the bag of $z$ to the vertices, removing the hyperedge $e$, and adding a hyperedge for every adhesion of a child of $z$. 
We now verify that this process is a special case of the replacement in Lemma~\ref{lem:thin-persist}.
Indeed, if we define hypergraph $K=\hyptorso(t_z,\set{z})$, where $t_z$ is the subtree of $t$ rooted at $z$, then the network of $Z\cup \set{z}$ is obtained from the network of $Z$ by replacing $e$ by $K$.
Observe that $K$ has at most $k+1$ vertices, has no hyperedge larger than $k$ due to the saneness of $t$, and is connected, again due to the saneness of $t$.
Hence Lemma~\ref{lem:thin-persist} ensures us that the network of $Z \cup \set z$ is also $k$-thin.
\end{proof}

\subsection{Decomposition into low pathwidth parts}\label{sec:global}

\newcommand{\quot}[1]{/_{\!#1}}
In this section  we finish the proof of Lemma~\ref{lem:low-pw-decomp}. We  heavily use the notation from Definition~\ref{def:tree-decomposition-terminology}.

Consider a tree decomposition $t$. 
For a distinguished set $X$ of nodes in the decomposition $t$, which is required to include all the roots of $t$, we define a new tree decomposition $t \quot X$ of the same underlying graph as follows.
The nodes of $t\quot X$ are $X$.
For any node of $t$ that is not in $X$, assign it to its closest ancestor that belongs to $X$, i.e., the ancestor from $X$ for which there is no other node from $X$ on the unique path between the node and the ancestor.
Then the $t \quot X$-bag of a node $x\in X$ is the union of the $t$-bags of the nodes assigned to $x$, plus the $t$-bag of $x$ itself.

Lemma~\ref{lem:low-pw-decomp} will be obtained by taking any sane tree decomposition, and applying the following lemma to every connected component.  
Condition~\eqref{p:pwbound-gl} of Lemma~\ref{lem:low-pw-decomp} will follow from Lemma~\ref{lem:seven-invariant}\eqref{p:seven-bagsize}, 
while condition~\eqref{p:cngbound-gl} of Lemma~\ref{lem:low-pw-decomp} will follow from Lemma~\ref{lem:small-cong} proved at the end of this section.

\begin{figure}[htbp!]
        \centering
                \def\svgwidth{0.7\columnwidth}
                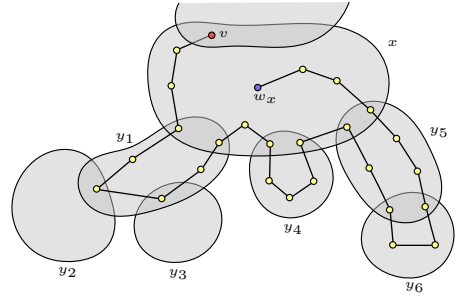
\caption{Example path in  $\Pp_{x}$. This path contributes to the loads of nodes $y_1$, $y_4$, $y_5$, and $y_6$, but not of $x$, $y_2$, and $y_3$.}\label{fig:normal}
\end{figure}

\begin{lemma}\label{lem:seven-invariant}
Let $t$ be a width $k$ sane tree decomposition of a connected graph $G$. 
Then one can find a set of nodes $X$ in $t$, which includes the root of $t$, and families of paths $\set{\Pp_x}_{x \in X}$, such that every $x \in X$ satisfies:
	\begin{enumerate}[(a)]
		\item \label{p:seven-bagsize} The $t \quot X$-marginal graph of $x$ has  pathwidth at most $2k+1$.
		\item \label{p:seven-normal} Every element of $\Pp_x$ is a path in $G$ that satisfies:
		\begin{enumerate}[(i)]
			\item\label{p:normal-down} except for its endpoints, the path  visits only vertices from the $t$-component of $x$;
			\item\label{p:normal-interval} if  $y \in X$  is a strict descendant of $x$, then restricting the path to the $t$-component of $y$ yields an interval in the path.
		\end{enumerate}
		\item \label{p:seven-guidance} All paths in $\Pp_x$ have the same source,  which  belongs to the $t$-margin of $x$. 
		Conversely, each vertex of the $t$-adhesion of $x$ is a target of some path from $\Pp_x$.
		\item \label{p:seven-load} The following set of paths has size at most $2k^3$:
		\begin{align*}
\mathrm{load}_x \eqdef \{P \in \Pp_{y} : \mbox{$y \in X$ is a strict ancestor of $x$ and}\\ \qquad \qquad \qquad \mbox{$P$ intersects the  $t$-component of $x$}\}.
		\end{align*}
	\end{enumerate}
\end{lemma}
\begin{proof}
Figure~\ref{fig:normal} illustrates the notions used in the lemma.  %
%
We prove  the following strengthening of the lemma, with sufficient parameters to be proved by induction.
\begin{itemize}
	\item[($\star$)] 	Let $x_0$ be a node of $t$, let $I$ be a set of size at most $2k^3$, and let
	\begin{align*}
		\set{(u_i,v_i)}_{i \in I} 
	\end{align*}
	be a set of node pairs from the adhesion in $x_0$, with possible repetitions. One can find a set $X 
\ni x_0$ of nodes in the subtree of $x_0$, sets  $\set{\Pp_x}_{x \in X}$ and a set of paths $\set{Q_i}_{i \in I}$ such that
	\begin{enumerate}[(A)]
		\item \label{p:star-requests} for every $i \in I$, the path $Q_i$ goes from $u_i$ to $v_i$ and satisfies conditions 
		                              (\ref{p:seven-normal}:i) and (\ref{p:seven-normal}:ii) in the lemma with respect to $x_0$;
		\item \label{p:star-seven} every $x \in X$ satisfies conditions (\ref{p:seven-bagsize})-(\ref{p:seven-guidance}) in the lemma and 
		                           the following variant of~\eqref{p:seven-load}: the size of $\mathrm{load}_x$ is at most 
		\begin{align*}
		2k^3 - |\set{i\in I\colon \mbox{path $Q_i$ intersects the $t$-component of $x$}}|.
		\end{align*}
	\end{enumerate}
\end{itemize}

The lemma is a special case of ($\star$), by choosing  $x_0$ to be the root of the tree decomposition $t$, and choosing $I$ to be  empty. 
It remains to prove  ($\star$), which is done  by induction on the number of nodes in the subtree of $x_0$. 

Let $x_0$ and $\set{(u_i,v_i)}_{i \in I}$ be as in ($\star$). 
Choose  $(u,v)$ so that 
\begin{align*}
	I_0 = \set{i \in I : (u_i,v_i)=(u,v)}
\end{align*}
has maximum size. If $I$ is empty,  choose $I_0$ to be empty and leave $(u,v)$ undefined, this  corner case will be considered separately.
Since each candidate for $(u,v)$ is in the adhesion of $x_0$, and adhesions have size bounded by $k$ (due to the saneness of $t$),  
it follows that  $I_0$ has size at least $|I|/k^2$. 
	
Define $t_0$ to be the subtree of $t$ rooted in $x_0$. Let  $t'_0$ be obtained from $t_0$ by removing from each bag those vertices of the adhesion of $x_0$ that are neither $u$ nor $v$. 
It is easy to see that both $t_0$ and $t'_0$ are sane tree decompositions. 
If $u,v$ are defined, apply  Lemma~\ref{lem:local-decomp} to $t_0'$ with distinguished vertices $u,v$, yielding a prefix $Z$ and two paths in the hypergraph $\hyptorso(t_0',Z)$; call these paths $P^1$ and $P^2$. 
Since $t_0$ and $t_0'$ have the same nodes, and these nodes are a subset of the nodes of $t$, we can view $Z$ as a connected set of nodes in each of these tree decompositions.  
In case $I$ is empty and $(u,v)$ is undefined, we choose $Z$ to be the singleton of the root of $t$.
%

Let $J$  be the disjoint union of $I$ and the $t$-adhesion of $x_0$. For $i \in J$, define a path $R_i$ in $\hyptorso(t_0,Z)$ as follows:
\begin{itemize}
\item For $i \in I-I_0$, define $R_i$ to be a path from $u_i$ to $v_i$, which does not visit the $t$-adhesion of $x_0$ except for its endpoints. Such a path exists by the saneness of $t$.
\item Split the set $I_0$ into two parts, with the first part having half the size (rounded up). 
      For $i$ in the first part, define $R_i$ to be $P^1$, and for $i$ in the second part, define $R_i$ to be $P^2$.
\item For the remaining $i$, i.e.~those from the $t$-adhesion of $x_0$, do the following. 
      Independently of $i$, choose some vertex $w$ in the $t$-margin of $x_0$, which is the same as the $t_0$-margin of $x_0$. 
      This is possible because margins are nonempty in a sane decomposition.  
      For each $i$ in the adhesion of $x_0$, define $R_i$ to be a path from $w$ to $i$, which does not visit the adhesion of $x_0$ except for its endpoints. 
      Such a path exists by the saneness of $t$.
\end{itemize}

By definition, every hyperedge in  $\hyptorso(t_0,Z)$ is an edge of $G$ or corresponds to an adhesion of some  $z \in \sucnodes$. Let  $z \in \sucnodes$. Define  $I^z$ to be those $i \in J$ for which path $R_i$ uses the hyperedge corresponding to the adhesion of $z$.
The following claim is the key step in this lemma, allowing us to apply the induction assumption.
\begin{claim}\label{claim:halved}
For every $z \in \sucnodes$, the set $I^z$ has size at most $2k^3$.
\end{claim}
\begin{proof}
The claim is trivial in the case when $I=\emptyset$ and $(u,v)$ is undefined, $I_z$ is a subset of $J$, which in this case has size at most $k$.
In the following, we focus on the general case when $I\neq \emptyset$.

Let $e$ be the hyperedge of $\hyptorso(t_0,Z)$ that corresponds to the adhesion of node $z$. 
To bound the size of $I_z$, we count the number of elements $i\in J$ for which the path $R_i$ traverses the hyperedge $e$.

First, there are at most $k$ elements of $J$ that originate from the $t$-adhesion of $x_0$. The remainder of $J$, being simply $I$, can be partitioned into $I_0$ and $I-I_0$.
Among paths $R_i$ of $i\in I_0$, at most half (rounded up) can traverse $e$. 
This is because $e$ cannot be traversed simultaneously by $P^1$ and by $P^2$, because it is not an edge of the underlying graph $G$.
We conclude that the size of $I^z$ can be bounded as follows:
\begin{align*}
|I_z| & \leq k+|I-I_0|+\lceil|I_0|/2\rceil\\
& = k+|I|-\lfloor |I_0|/2\rfloor\\
& \leq k+|I|-\lfloor |I|/2k^2\rfloor,
\end{align*}
where the last inequality follows from $|I_0|\geq |I|/k^2$. 
We now use the fact that $|I|\leq 2k^3$ and that the function $a\mapsto a-\lfloor a/2k^2\rfloor$ is non-decreasing to conclude that
\begin{align*}
|I_z| & \leq k+2k^3-k=2k^3.
\end{align*}\cqed\end{proof}

The informal reason for the claim is that we have used the two paths $P^1$ and $P^2$, and at most one of them visits the hyperedge corresponding to the adhesion of $z$.
Since $I_0$ constitutes a $1/k^2$-fraction of $I$, and $I_0$ is split in halves with respect to using $P^1$ or $P^2$, we see that around $1/2k^2$-fraction of paths $R_i$ for $i\in I$ avoid the adhesion of $z$.
This is enough to amortize for the new paths $R_i$ for $i$ from the $t$-adhesion of $x_0$.

 For $i \in I^z$, let $u_i^z$
 be the vertex used by the path $R_i$ immediately before the hyperedge corresponding to the adhesion of $z$, and let $v^z_i$  be the vertex used immediately after. 
 Take the vertex $z$, which is a proper descendant of $x_0$, and the family 
\begin{align*}
	\set{(u^z_i,v^z_i)}_{i \in I^z}.
\end{align*}
and apply to it the induction assumption of ($\star$), yielding
\begin{align*}
	X^z \qquad  \set{\Pp^z_x}_{x \in X^z} \qquad \set{Q^z_i}_{i \in I^z}.
\end{align*}
For $i \in J$, define $Q_i$ to be following path in $G$. Take the path $R_i$, which might use hyperedges corresponding to adhesions from $\sucnodes$, 
and  for every  $z \in \sucnodes$ replace the hyperedge corresponding to the adhesion of $z$, if it is used, by the path $Q^z_i$. 
The following claim follows directly from the construction and the induction assumption.
\begin{claim}\label{cl:pipi}
For every $i \in J$, the path $Q_i$ has the same source and target as $R_i$ and satisfies conditions~(\ref{p:seven-normal}:i) and (\ref{p:seven-normal}:ii) from the lemma.
\end{claim}

We now complete the proof of ($\star$). Define 
\begin{align*}
	X = \set{x_0} \cup \bigcup_{z \in \sucnodes} X^z.
\end{align*}
Note how the above union is actually a partition. Define
\begin{align*}
\Pp_x = 	\begin{cases}
		\set{Q_i : \mbox{$i$ in the $t$-adhesion of $x_0$}} & \mbox{when $x=x_0$}\\
		\Pp^z_{x} & \mbox{when $x \in X^z$}.
	\end{cases}
\end{align*}
Finally, define $\set{Q_i}_{i \in I}$ to be the restriction of the previously defined  family $\set{Q_i}_{i \in J}$ to the smaller indexing set $I \subseteq J$.

Let us check that conditions~(\ref{p:star-requests}) and~(\ref{p:star-seven}) in the conclusion of ($\star$) are satisfied by the above choices. 
Condition~(\ref{p:star-requests}) is satisfied thanks to Claim~\ref{cl:pipi}. 
For $x \in X - \set {x_0}$, we check condition~(\ref{p:star-seven})  in the following claim, which follows from the induction assumption.

\begin{claim}\label{cl:check-load}
Condition~(\ref{p:star-seven}) in  ($\star$) holds for each $x\in X - \set {x_0}$.
\end{claim}
\begin{proof}
Let $z$ be the node of $\sucnodes$ such that $x$ belongs to the subtree rooted at some $z$. 
From the induction assumption we infer that $\mathrm{load}_x$, when computed with respect the family $\set{\Pp^z_x}_{x \in X^z}$, has size at most
\begin{align*}
2k^3 - |\set{i\in I^z\colon \mbox{path $Q^z_i$ intersects the $t$-component of $x$}}|.
\end{align*}
Let us denote the set whose cardinality is subtracted above by $I^z_x$.

When computing $\mathrm{load}_x$ with respect to the whole family $\set{\Pp_x}_{x \in X}$, we need to also take into account the contribution from paths $Q_i$ for $i$ in the adhesion of $x_0$.
More precisely, this contribution is equal to the cardinality of the set $I'$ of all vertices $i$ from the $t$-adhesion of $x_0$, for which the path $Q_i$ intersects the $t$-component of $x$.

Let us denote by $I^{x_0}_x$ the set of all indices $i\in I$, for which path $Q_i$ intersects the $t$-component of $x$. Since $J$ is the disjoint union of $I$ and the adhesion of $x_0$, 
and the elements of $I^z$ are in one-to-one correspondence with the elements $i\in J$ for which path $Q_i$ intersects the $t$-component of $z$ 
(which in particular happens if the path intersects the $t$-component of $x$), we infer that
$$|I^{z}_x|=|I'|+|I^{x_0}_x|.$$
Hence, $\mathrm{load}_x$ computed with respect to the whole family $\set{\Pp_x}_{x \in X}$ is upper bounded by
\begin{align*}
2k^3 - |I^{z}_x| + |I'|=2k^3-|I^{x_0}_x|,
\end{align*}
as requested in condition~(\ref{p:star-seven}) in the conclusion of ($\star$).\cqed\end{proof}

It remains to check condition~(\ref{p:star-seven}) for $x_0$, i.e.~to check that $x_0$ satisfies conditions~(\ref{p:seven-bagsize})-(\ref{p:seven-guidance}) and the variant of (\ref{p:seven-load}). 
For (\ref{p:seven-bagsize}), we observe that every path decomposition for $\hyptorso(t'_0,Z)$ is also a path decomposition of  $t \quot X$-margin of $x_0$, and therefore the latter graph has pathwidth at most $2k+1$, by the conclusions of Lemma~\ref{lem:local-decomp}. (In the corner case when $I$ was empty, the $t \quot X$-margin of $x_0$ is the same as the $t$-margin, and therefore has size at most $k$.) For~(\ref{p:seven-normal}), we use Claim~\ref{cl:pipi}. Condition~(\ref{p:seven-guidance}) follows by construction. 
Finally, condition~(\ref{p:seven-normal}) holds vacuously, because $x_0$ has no strict ancestors in $X$.
\end{proof}

To finish the proof of Lemma~\ref{lem:low-pw-decomp}, we take $\tb=t \quot X$ for $X$ yielded by Lemma~\ref{lem:seven-invariant}, and we are left with verifying that the adhesions of $\tb$ 
can be captured by a guidance system that uses few colors. This verification is given in the following lemma.

\begin{lemma}\label{lem:small-cong}
	Let $t$ and $X$ be as in Lemma~\ref{lem:seven-invariant}. The adhesions of the tree decomposition $t \quot X$ can be captured by a guidance system that is colorable with $4k^3+2k$ colors.
\end{lemma}
\begin{proof}
	In the proof we will use the notation $\comp$, $\bag$, and $\sep$ explained in the beginning of Appendix~\ref{app:guidance}. These operators respectively denote the component, the bag, and the adhesion at some node,
	and they will always be applied to the decomposition $\tb$.

	Let $M$ be the set of all pairs $(x,v)$ such that $x$ is a node of $\tb$ and $v$ is in the adhesion of $x$.
	For each node $x$ of $\tb$, let $w_x$ be the common start of all the paths from $\Pp_x$; recall that $w_x$ belongs to the $\tb$-margin of $x$.
	For each $(x,v)\in M$, let $P_{x,v}$ be a path from $\Pp_x$ that goes from $w_x$ to $v$.
	
	We will say that two pairs $(x_1,v_1),(x_2,v_2)\in M$ are {\em{in conflict}} if the following conditions are satisfied:
	$$v_1\neq v_2\qquad \textrm{and}\qquad \textrm{$P_{x_1,v_1}$ and $P_{x_2,v_2}$ intersect.}$$
	Intuitively, the conflict relation encodes that paths $P_{x_1,v_1}$ and $P_{x_2,v_2}$ have to be realized using different colors in a guidance system that captures the adhesions of $\tb$. 
	Note that two pairs with the same target vertex $v$ are never in conflict, even if they share some other vertices. This detail will be crucial for the proof of the claim.

	Let us define {\em{conflict graph}} $D$ with vertex set $M$ where the edge set encodes the relation of being in conflict.
	We now prove that $D$ admits a proper coloring using at most $4k^3+2k$ colors. 
	For this, we show that pairs from $M$ admit an ordering $\prec$ such that every pair $(x,v)$ has only at most $4k^3+2k-1$ conflicting pairs that are smaller in $\prec$.
	Then, a greedy coloring procedure, which scans $M$ in the $\prec$ order and assigns an arbitrary free color to each element of $M$, 
	yields a coloring of $D$ with at most $4k^3+2k$ colors\footnote{This argument is equivalent to saying that $D$ is $(4k^3+2k-1)$-degenerate, and hence $(4k^3+2k)$-colorable.}.

	As $\prec$ we take an arbitrary ordering of $M$ that respects the top-down order in $\tb$.
	That is, it has the following property: whenever $x$ is a descendant of $x'$, then $(x',v')\prec (x,v)$ for each $v'\in \sep(x')$ and $v\in \sep(x)$.
	Fix any pair $(x,v)\in M$; we now examine how many pairs $(x',v')\prec (x,v)$ can be in conflict with $(x,v)$.

	First, there are at most $k-1$ other pairs of $M$ where $x'=x$. 
	Keeping these pairs in mind, from now on we consider only pairs where $x'\neq x$. 
	Since $(x',v')$ is in conflict with $(x,v)$, we have that $v\neq v'$ and $P_{x,v}$ and $P_{x',v'}$ intersect.

	Second, from condition (\ref{p:seven-normal}:i) of Lemma~\ref{lem:seven-invariant} we infer that $P_{x,v}$ is entirely contained in $\comp(x)$ apart from its endpoint $v$. 
	Similarly, $P_{x',v'}$ is entirely contained in $\comp(x')$ apart from its endpoint $v'$. 
	Since $v\neq v'$, it must hold that $x$ and $x'$ are in ancestor-descendant relation, because otherwise $P_{x,v}$ and $P_{x',v'}$ could not intersect. 
	Since $(x',v')\prec (x,v)$, we infer that $x'$ is an ancestor of $x$.

	Suppose now that $P_{x',v'}$ intersects $\comp(x)$. Since $x'$ is an ancestor of $x$, it follows that $P_{x',v'}$ belongs to $\mathrm{load}_x$. 
	Since the size of $\mathrm{load}_x$ is at most $2k^3$, due to condition \eqref{p:seven-load} of Lemma~\ref{lem:seven-invariant}, the number of such pairs $(x',v')$ is upper bounded by $2k^3$.
	Keeping these pairs in mind, from now on we investigate only paths $P_{x',v'}$ that are disjoint with $\comp(x)$. 
	Since $P_{x,v}$ is entirely contained in $\comp(x)$ apart from its endpoint $v$, 
	the only remaining way $(x,v)$ and $(x',v')$ can be in conflict is when $P_{x',v'}$ traverses $v$ but $v\neq v'$.

	Let $y$ be the topmost node in $\tb$ such that $v\in \bag(y)$. 
	Clearly $y$ is also a (strict) ancestor of $x$, because $v\in \sep(x)$. 
	Thus, $y$ and $x'$ are also in ancestor-descendant relation. 
	Suppose for a moment that $x'$ is a strict descendant of $y$.
	As $v\in \bag(x)\cap \bag(y)$, we have that $v\in \sep(x')$. 
	However, this is a contradiction with condition (\ref{p:seven-normal}:i) of Lemma~\ref{lem:seven-invariant} for the path $P_{x',v'}$, 
	as this path traverses $v$ but avoids $\sep(x')$ apart from its endpoint $v'$, which is different from $v$.

	This proves that either $x'=y$ or $x'$ is an ancestor of $y$.
	The number of pairs $(x',v')$ with $x'=y$ is at most $k$.
	Keeping these pairs in mind, from now on we concentrate on the remaining case when $x'\neq y$.
	Since $y$ was chosen to be the top-most bag containing $v$, we have that $v\in \comp(y)$.
	As $v$ is traversed by $P_{x',v'}$, we infer that this path belongs to $\mathrm{load}_y$. 
	However, by condition \eqref{p:seven-load} of Lemma~\ref{lem:seven-invariant} the size of $\mathrm{load}_y$ is at most $2k^3$, 
	so the number of such pairs $(x',v')$ is at most $2k^3$.

	Thus, in total we have found at most $(k-1)+2k^3+k+2k^3=4k^3+2k-1$ pairs $(x',v')$ that are smaller than $(x,v)$ in the ordering $\prec$ and are in conflict with $(x,v)$. 
	As discussed before, this implies that $D$ admits a proper coloring $\phi$ with $4k^3+2k$ colors, so that any conflicting pairs receive different colors.

	We now define a guidance system $\guid$ over $G$ that captures the family of adhesions of $\tb$. 
	For each color $c$, let $G_c$ be the union of paths $P_{x,v}$ for which $\phi(x,v)=c$.
	As $\phi$ is a proper coloring of $D$, pairs $(x,v)$ mapped to the same color $c$ are pairwise not in conflict.
	Hence, all paths $P_{x,v}$ that are in the same connected component of $G_c$ actually need to have exactly the same target vertex $v$, as otherwise some two of them would be in conflict.
	For each connected component $C$ of $G_c$, let $v_C$ be this common target vertex.
	Construct $\guid$ as follows: for each color $c$ and each component $C$ of $G_c$, arbitrarily select any its arbitrary spanning tree, orient it toward $v_C$, and add it to $\guid$.
	Assigning color $c$ to each tree originating in $G_c$ yields a coloring of $\guid$ with $4k^3+2k$ colors.

	It is now easy to verify that $\guid$ captures all the adhesions of $\tb$. 
	Fix some node $x$ of $\tb$; we claim that $\sep(x)\subseteq \guid(w_x)$.
	Take any $v\in \sep(x)$, and let $c=\phi(x,v)$.
	Then path $P_{x,v}$ has been included in $G_c$, and in particular from the construction it follows that $v\in \guid(w_x)$.
\end{proof}

\section{Conclusions}\label{sec:conclusions}
There are two concrete open questions that arise from our work. 
Firstly, we believe that our main result, Theorem~\ref{thm:compute-tree-decomposition}, can be strengthened to the following statement:
the transduction yields a tree decomposition of the optimum width $k$, instead of approximate $f(k)$.
This would immediately lead to a stronger statement of Theorem~\ref{thm:courcelle-conjecture} (Courcelle's conjecture): 
the assumption of $k$-recognizability alone, instead of $k'$-recognizability for all $k'$, would imply that a property of graphs of treewidth $k$ is definable in counting \mso.

Our idea for proving the stronger statement is to take a closer look on the work of Bodlaender and Kloks~\cite{BodlaenderK96}, who gave a dynamic programming algorithm that, 
given a graph together with a tree decomposition of width $k'\geq k$, constructs, if possible, a tree decomposition of width $k$. 
A preliminary inspection of the proof gives hope that the algorithm can be translated into a deterministic \mso transduction that,
given a tree decomposition of width $k'$, outputs a tree decomposition of width $k$. Such a transduction could be then combined with Theorem~\ref{thm:compute-tree-decomposition} to yield the strengthening.

Secondly, in Section~\ref{lem:pathwidth-comb} we have proved that the guided treewidth of a graph is bounded in terms of its pathwidth.
It is natural to ask whether the same holds also for treewidth: does there exist a function $f$ such that
\begin{align*}
	\gtw(G)\leq f(\tw(G)) \qquad \mbox{for every graph $G$}?
\end{align*} 
Our current approach falls short of proving this conjecture. 
Intuitively, the main problem is that the combinatorial results of Lemmas~\ref{lem:pathwidth-comb} and~\ref{lem:low-pw-decomp} are combined at the (less restrictive)
level of \mso transductions in the proof of Theorem~\ref{thm:compute-tree-decomposition}, and we do not know how to combine them at the level of guidance systems.
If true, the conjecture would give a somewhat conceptually easier way to prove our main result.

\bibliographystyle{abbrv}
\bibliography{courcelle}

\appendix

\section{Proof of the main result (Theorem~\ref{thm:compute-tree-decomposition})}\label{sec:wrapup}
\newcommand{\Gt}{\widehat{G}}

\newcommand{\fibration}[1]{\mathsf{fibration}(#1)}
\newcommand{\fiber}[1]{\mathsf{fiber}_x}
In this section, we combine Lemmas~\ref{lem:compute-path-decomposition} and~\ref{lem:path-tree-decomposition} to get the main technical Theorem~\ref{thm:compute-tree-decomposition}.  
The proof strategy is to take a graph, apply Lemma~\ref{lem:path-tree-decomposition} to get a  tree decomposition with  marginal graphs of bounded pathwidth, 
then apply   Lemma~\ref{lem:compute-path-decomposition} to get a tree decomposition of each marginal graph, and then to combine these tree decompositions into a single tree decomposition.

In the proof we will need the following simple technical statement about sane tree decompositions. 
Its proof can be found in Appendix~\ref{app:guidance-long-proof} (see Claim~\ref{cl:exists-direct} therein).
\begin{claim}\label{cl:exists-direct-pre}
Suppose $t$ is a sane tree decomposition of a graph $G$, $y$ is a node of $t$, and $x$ is the parent of $y$. 
Then there exists a vertex of $G$ that is simultaneously in the adhesion of $y$ and in the margin of its parent $x$.
\end{claim}

\paragraph*{Nested decompositions.}  
A \emph{nested tree decomposition} consists of a sane tree decomposition $t$ plus a sane tree decomposition $t_x$ of the $t$-marginal graph of $x$ for every node $x$ of $t$.  
We use the name \emph{main decomposition} for $t$, and the name \emph{marginal decomposition} for any of the $t_x$. 
Note that by the saneness of $t$, the marginal graph of each node of $t$ is nonempty and connected, and therefore each marginal decomposition is a tree and not a forest. 

We model a nested tree decomposition as a logical structure as follows.
The structure is the disjoint union of the structures encoding the main decomposition and all the marginal decompositions.
Recall that each of these structures contains both the decomposition (as nodes connected by the $\Parent$ relation) and the underlying graph (using the standard encoding for graphs), which are linked together by the
$\Bag$ relation encoding the contents of the bags. 
Thus, the structure encoding a nested tree decomposition contains a copy of the whole graph, linked to the main decomposition $t$, as well as, 
for each node $x$ of $t$, a copy of the $t$-marginal graph of $x$, linked to its marginal decomposition $t_x$.
These will be called the {\em{underlying graphs}} of $t$ and $t_x$, respectively.

In the encoding of a nested tree decomposition, we assume that there is an additional binary predicate $\mathsf{same}$ that selects pairs $(v,v')$ such that 
$v$ and $v'$ are copies of the same vertex: one in the underlying graph of $t$, and the second in the underlying graph of $t_x$, where the vertex belongs to the $t$-margin of $x$.
We also assume that there is a binary predicate $\mathsf{marginal}(y,x)$ that selects pairs $(y,x)$ such that $y$ is a node of the marginal decomposition $t_x$.
Note that $\mathsf{marginal}$ is a function on the union of the node sets of the marginal decompositions.

Define the \emph{width} of a nested tree decomposition to be equal to 
\begin{gather*}
\textrm{(maximum width among marginal decompositions $t_x$)}\\
+\\
\textrm{(maximum size of an adhesion of $t$)}.
\end{gather*}
The following lemma shows that a deterministic \mso transduction can flatten a nested tree decomposition.

\newcommand{\outnode}{\mathsf{out}}
\begin{lemma}\label{lem:flattening}
There exists a deterministic \mso transduction, whose domain is the set of nested tree decompositions, 
which transforms a nested tree decomposition into a tree decomposition with the same underlying graph. 
If the input has width at most $k$, then the output also has width at most $k$.
\end{lemma}
Note that in the above lemma, the transduction does not fix the width $k$, i.e.~there is a single transduction that works for all $k$.
\begin{proof}
Consider a nested tree decomposition of a graph $G$, consisting of the main decomposition $t$ and marginal decompositions $t_x$, for nodes $x$ of $t$.   
Below we define its \emph{flattening}, which is going to be the output, to be the following tree decomposition $s$ of $G$. 
	\begin{itemize}
		\item {\it Nodes.} Nodes of the flattening $s$ are pairs $(x,y)$ such that $x$ is a node of the main decomposition and $y$ is a node in $t_x$. 
		\item {\it Bags.}  Let  $(x,y)$ be a node of the flattening $s$.  The $s$-bag of $(x,y)$  is  the union of the $t$-adhesion of $x$ and the  $t_x$-bag of $y$.
		
		\item {\it Parents.} Let $(x,y)$ be a node in the flattening $s$. Let $x',y'$ be the parents of $x,y$ in their respective tree decompositions, both of which might be undefined. 
		If $y'$ is defined, then  the parent of $(x,y)$ is $(x,y')$. If $x'$ is undefined, then $(x,y)$ is a root of the flattening $s$. 
		Otherwise, if $y'$ is undefined but $x'$ is defined, we examine set $A$ defined as the intersection of the $t$-adhesion of $x$ with the $t$-margin of $x'$. 
		By Claim~\ref{cl:exists-direct-pre}, $A$ is nonempty, and moreover it is a clique in the $t$-marginal graph of $x'$. 
		Since $t_{x'}$ is a tree decomposition of this $t$-marginal graph, the set of nodes of $t_{x'}$ whose bags contain $A$ is nonempty and connected in $t_{x'}$.
		Let $z$ be the lowest among these nodes; in other words, $z$ is the lowest common ancestor of all the nodes whose bags contain $A$, which is also a node with this property.
		Then the parent of $(x,y)$ is defined to be $(x',z)$.
	\end{itemize}
	It is not difficult to check that the flattening described above  is indeed a tree decomposition. 
	Note that the definition of the width of a nested tree decomposition is chosen in such a way that, provided the input nested tree decomposition has width at most $k$, 
	then the width of the output tree decomposition is also at most $k$.
	
We now shortly argue that the construction above can be implemented by means of a deterministic \mso transduction.
A node $(x,y)$ of the output $t$ is interpreted in the node $y$ of $t_x$ on the input; note that the node $x$ can be uniquely recovered from $y$ using predicate $\mathsf{marginal}$.
Then it is straightforward to implement all the definitions above using interpretation.
\end{proof}

\newcommand{\htt}{\mathfrak T}

Theorem~\ref{thm:compute-tree-decomposition} is an immediate corollary of Lemma~\ref{lem:flattening} and the following lemma, 
which is where we use Lemmas~\ref{lem:compute-path-decomposition} and~\ref{lem:path-tree-decomposition}.
\begin{lemma}\label{lem:to-nesting}
	For every $k\in \Nats$ there is some $k'\in \Nats$ and  an \mso transduction, whose domain is graphs,  such that:
	\begin{itemize}
		\item every output is a nested tree decomposition of the input;
		\item if the input has treewidth at most $k$, then  some output has width at most $k'$.
	\end{itemize}
\end{lemma}

\begin{proof}
Suppose we are given some graph $G$ together with a coloring $\phi$ of its vertices using some finite set of colors.
We will say that a tree decomposition $t$ of $G$ is {\em{colorful under $\phi$}} if every adhesion of $t$ has all its vertices colored differently in $\phi$.
For some $\ell\in \Nats$, by an {\em{$\ell$-colored}} tree decomposition we mean a tree decomposition given as a logical structure (i.e. with the underlying graph encoded),
where the underlying graph $G$ is additionally supplied with some vertex coloring which uses $\ell$ colors and under which $t$ is colorful. 
The coloring is given as $\ell$ unary predicates that partition the vertex set.
Obviously, an $\ell$-colored tree decomposition has no adhesion larger than $\ell$.

Let $t$ be a tree decomposition, given as a logical structure (i.e. with the underlying graph $G$ encoded).  
Define the logical structure $\htt$ to be the disjoint union of $t$ and all of its marginal graphs, 
together with a binary predicate $\mathsf{same}$ that selects all pairs $(v,v')$ such that $v$ and $v'$ are copies of the same vertex:
$v$ belongs to the graph $G$ (present in the logical structure encoding $t$), and $v'$ is its copy in the unique marginal graph that contains $v$.

\begin{claim}\label{cl:cant-stop-the-rock}
Let $\ell\in \Nats$. 
There is an \mso transduction that implements the transformation $t \mapsto \htt$ for $\ell$-colored tree decompositions $t$.
\end{claim}
\begin{proof}
We first create $2+\binom{\ell}{2}$ copies of $t$ (called further layers), and then shape these layers into the output structure $\htt$ using interpretation in a manner described as follows.
The first layer is simply the copy of $t$ in the output. 
The second layer is used to define the vertices of the marginal graphs, and the edges in the marginal graphs that correspond to the edges in $G$. The \mso-definability of
the construction that carves out the marginal graphs from a copy of $t$ is straightforward. Here, we can also define predicate $\mathsf{same}$.

Let us index the remaining $\binom{\ell}{2}$ layers with 2-element subsets of $\set{1,\ldots,\ell}$; these layers will be used to produce the additional edges that are added in the construction of the marginal graphs.
More precisely, for every node $x$ in $t$ and every its child $y$, we need to add the at most $\binom{\ell}{2}$ edges between vertices of the intersection of the adhesion of $y$ and the margin of $x$.
Let us call this intersection $A_{x,y}$.
For the copy of $y$ from the layer indexed with $\{i,j\}$, we check whether $A_{x,y}$ contains both a vertex $u$ colored with color $i$, and a vertex $v$ colored with color $j$, and moreover
vertices $u$ and $v$ are not adjacent in $G$. If this is not the case, we remove the copy from the structure. Otherwise, we turn the copy into an edge between $u$ and $v$ in the marginal graph of $x$.
Note that the copies of $u$ and $v$ in this marginal graph can be retrieved using predicate $\mathsf{same}$, which we have already defined.
Observe that, in this manner, an edge $uv$ could have been added multiple times to the structure, in case it is contained in multiple adhesions of $t$. 
However, we can remove the multiplicities by guessing one duplicate of each edge to preserve, and removing all the others.
\end{proof}

We now complete the proof of the lemma. Suppose that we are given on input a graph $G$.
For $k$, apply  the \mso transduction  from Lemma~\ref{lem:path-tree-decomposition}, yielding a tree decomposition $t$. 
If the input graph $G$ had treewidth at most $k$, then there is at least one output $t$ for which all marginal graphs in $t$ have pathwidth at most $2k+1$, and
the family of adhesions of $t$ can be captured by a $(4k^3+2k)$-colorable guidance system.

\begin{claim}\label{cl:small-chromatic}
There exists a coloring of the vertex set of $G$ using at most $4k^3+4k+2$ colors such that every adhesion of $t$ has all vertices colored differently.
\end{claim}
\begin{proof}
Since the adhesions of $t$ can be captured by a guidance system that can be colored using $4k^3+2k$ colors, this means that all adhesions of $t$ have at most this size.
This, together with the bound on the pathwidth of each marginal graph, shows that $G$ admits a nested tree decomposition of width at most $4k^3+4k+1$. 
This nested tree decomposition remains a valid nested tree decomposition of the same width even if we consider a graph $G'$ obtained from $G$ by turning every adhesion of $t$ into a clique.
By applying the combinatorial construction of Lemma~\ref{lem:flattening} to this nested tree decomposition with the underlying graph $G'$, 
we infer $G'$ admits a tree decomposition of width at most $4k^3+4k+1$, i.e.,
$$\tw(G')\leq 4k^3+4k+1.$$
The result follows from a known fact that any graph of treewidth $\ell$ admits a proper coloring using $(\ell+1)$ colors, and the property that in $G'$ all adhesions of $t$ are cliques.
\cqed\end{proof}

Setting $\ell=4k^3+4k+2$, we guess a coloring of the vertex set of $G$ with $\ell$ colors that has the property stated in Claim~\ref{cl:small-chromatic}.
Then we can enrich the tree decomposition $t$ with the guessed coloring, thus obtaining an $\ell$-colored tree decomposition.
To this $\ell$-colored tree decomposition we can apply the \mso transduction from Claim~\ref{cl:cant-stop-the-rock}, which constructs the marginal graphs.

For now leave the tree decomposition $t$ alone, and to the union of the marginal graphs apply the \mso transduction from Lemma~\ref{lem:compute-path-decomposition}. 
Since the pathwidth of the union of these graphs is bounded by $2k+1$, this yields a tree decomposition of the marginal graphs with width bounded by a function of $k$.

So far, the result is two tree decompositions: the tree decomposition $t$ and a tree decomposition $s$ of the union of the marginal graphs.
Tree decomposition $s$ can be actually assumed to be sane; see the proof of Lemma~\ref{lem:compute-path-decomposition} in Appendix~\ref{app:pathwidth}.
Since a sane tree decomposition of a graph contains one tree per each connected component, we see that
in $s$ there is one tree for each marginal graph, because all marginal graphs are connected. This tree is a sane tree decomposition of this marginal graph.
Hence, $s$ together with the marginal graphs actually {\em{is}} a disjoint union of logical structures representing sane tree decompositions of the marginal graphs.

To fully define a nested tree decomposition, 
we are left with defining the predicate $\mathsf{marginal}(y,x)$ that associates each node $y$ of the tree decomposition $t_x$ of the marginal graph of $x$, with the node $x$.
This is, however, straightforward: if $u$ is any vertex of the margin of $y$ in $t_x$ (which exists by the saneness of $t_x$), then $x$ is equal to the unique node of $t$ whose margin contains $u$.
Note that here we use the already defined predicate $\mathsf{same}$ to relate the vertices of the whole graph $G$ with their copies in the marginal graphs.

Since the adhesions of $t$ have sizes at most $4k^3+2k$, because they can be captured by a guidance system using this many colors, it follows that the width of the obtained nested tree decompositions
can be bounded by a function of $k$.
\end{proof}

\section{Omitted proofs from Section~\ref{sec:overview}}\label{app:overview}
\begin{lemma}[Backwards Translation Theorem]\label{lem:pull-back}
Suppose $\Ii$ is an \mso transduction with input vocabulary $\Sigma$ and output $\Gamma$, and let $\psi$ be an \mso sentence over vocabulary $\Gamma$.
Then there exists an \mso sentence $\varphi$ over vocabulary $\Sigma$ such that $\varphi$ is true in exactly those structure $\mathfrak A$ over $\Sigma$, for which $\Ii(\mathfrak A)$ contains at least one structure
satisfying $\psi$.
\end{lemma}
\begin{proof}[Proof sketch]
It suffices to verify the statement for the three basic types of operations. 

For copying, we replace every quantification in $\psi$ of some subset $X$ with quantification of $n$ sets $X_1,X_2,\ldots,X_k$, 
which represent the intersections of $X$ with consecutive layers of the copied universe.
Similarly, individual quantification in $\psi$ of an element $e$ can be replaced by quantification of this element and the layer it belongs to, 
where the latter is implemented as a constant-size conjunction or disjunction.
It is then straightforward to adapt atomic formulas in $\psi$.

For coloring, formula $\varphi$ simply existentially quantifies predicates $X_1,\ldots,X_k$, and then uses $\psi$.

For interpretation, formula $\varphi$ is the conjunction of formulas $\domain$ and $\psi'$, where $\psi'$ is derived from $\psi$ as follows:
every predicate $R$ is replaced with the corresponding formula $\varphi_R$, and quantifications are relativized to the restricted universe using formula $\univformula$.
\end{proof}

\begin{proof}[Proof of Lemma~\ref{lem:courcelle}]
The following proof is an adaptation of the proof of Theorem 4.8, claim (2), from~\cite{Courcelle91a}. 
Our strategy is to reduce the case of graphs supplied with tree decompositions to the case of forests labelled by some finite alphabet. 
Then we use a result of Courcelle~\cite{Courcelle90}, which proves an analogous claim for labelled forests.

Let $t$ be the tree decomposition of the graph $G$ given on the input. 
Recall that the input structure encodes both the graph $G$ and its tree decomposition $t$, and the width of $t$ is bounded by $k$.

Without loss of generality, we can assume that no neighboring nodes in $t$ have equal bags, since otherwise the edge between these two nodes could be contracted.
Indeed, it is straightforward to see that such a contraction can be simultaneously applied to all such pairs of neighboring nodes by means of an \mso transduction; 
more precisely, the contracted decomposition can be interpreted in the input one. 
Hence, from now on we assume that $t$ has the aforementioned property, and in particular all the adhesions in $t$ (i.e. intersections of neighboring bags) have sizes at most $k$.

We now apply an \mso transduction $\cal C$ that does the following:
guess a coloring of the vertex set of $G$ using $k+1$ colors such that every two vertices that appear together in some bag of $t$ receive different colors.
It is a known fact about colorings of graphs of bounded treewidth that, provided the width of $t$ is at most $k$, such a coloring exists.
Thus, from now on we can assume that the input graph is supplied with a $(k+1)$-coloring having the property stated above.

Define $\Sigma$ to be the set of pairs $(H,X)$ with the following property: $H$ is a graph whose vertex set is a subset of $\set{1,\ldots,k+1}$, and $X$ is a subset of the vertices of $H$ that has size at most $k$.
Clearly, $\Sigma$ is a finite set of size that is a function of $k$ only.
We now apply an \mso transduction $\cal D$ that labels each node of the tree decomposition $t$ with the subgraph induced by its bag, together with the adhesion of the node.
Formally, to each node $x$ of $t$ we associate the label $(H,X)\in \Sigma$, where
$H$ is the subgraph induced by the bag of $x$ in $G$, with vertices renamed according to their colors in the guessed $(k+1)$-coloring, and $X$ is the image of the adhesion of $x$ under the same renaming.
As usual, labels are encoded by adding a fresh unary predicates, one per label in  $\Sigma$; such a  predicate is true for exactly those nodes that are labelled with the element.
It is clear that such labelling can be computed by an \mso transduction.

The intuition now is that the labelling of the decomposition uniquely encodes the underlying graph $G$, and the encoding can be reversed provided there is no ``mismatch'' between the encodings of two neighboring bags.
We now formalize this intuition.

A tree $s$ with nodes labelled by $\Sigma$ is called {\em{well-formed}} if the following assertion holds for every node $y$ with parent $x$: 
if $x$ and $y$ are labelled with $(H_x,X_x)$ and $(H_y,X_y)$, respectively, then $X_y$ is a subset of the vertex set of $H_x$ and the graphs induced by $X_y$ in $H_x$ and $H_y$ are equal.
It is clear that the output of the interpretation $\cal D$ above is always well-formed.

A well-formed tree $s$ labelled with $\Sigma$ may be uniquely associated with a $(k+1)$-interface graph as follows. 
If $r$ is the root of $t$ and it is labelled by $(H,X)$, then the $(k+1)$-interface graph associated with $s$ will use $X$ for the names of the interfaces.
This interface graph is recursively defined as follows: 
\begin{itemize}
\item Define $\Hf_r$ to be the interface graph obtained from $H$ by taking the identity on the vertex set of $H$ as the interface mapping.
\item Compute the $(k+1)$-interface graphs $\Hf_1,\Hf_2,\ldots,\Hf_p$ associated with the subtrees rooted at the children of the root $r$.
\item Take the gluing of all the interface graphs defined above.
\item Forget (i.e. remove from the interface mapping) all the interfaces whose names do not belong to $X$.
\end{itemize}
The $(k+1)$-interface graph associated with a well-formed tree $s$ will be called the {\em{decoding}} of $s$.
It is clear the input graph $G$ is isomorphic to the graph obtained by considering any labelled forest output by the interpretation $\cal C\circ \cal D$, 
and taking the disjoint union of the decodings of the trees of this forest (each of these decodings has no interfaces).

We can define recognizability for properties of forests labelled with some finite alphabet as follows. 
First, by a {\em{context}} we mean a forest over the alphabet with one specified node called the {\em{hole}}, which is a leaf and has no label.
If $c$ is a context and $s$ is a tree, then we can obtain a forest by replacing the hole in $c$ with tree $s$ in a natural manner: we take the disjoint union of $c$ and $s$, remove the hole, 
and attach the root of $s$ as a child of the parent of the hole, in case it exists.
This operation, denoted by $c\circ t$, naturally leads to the definition of recognizability for properties of forests.
Namely, for some property $\Pi$ of labelled forests, we will say that two trees $s_1,s_2$ are {\em{$\Pi$-equivalent}} if, and only if for any context $c$ we have that
\begin{align*}
	c\circ s_1 \mbox{ satisfies } \Pi \qquad\mbox{iff} \qquad c\circ s_2\mbox{ satisfies $\Pi$}.
\end{align*}
We shall say that $\Pi$ is {\em{recognizable}} if the equivalence relation defined above has finitely many equivalence classes.

Now, given the property $\Pi$ of graphs, we define property $\Pi'$ of rooted forests labelled with $\Sigma$ as follows.
A forest $t$ satisfies $\Pi'$ if, and only if:
\begin{itemize}
\item Each tree of $s$ is well-formed, and the label of its root has the empty set on the second coordinate; and
\item The graph obtained by taking the disjoint union of the decodings of the trees of $s$ satisfies $\Pi$. 
\end{itemize}
We now deduce that the recognizability of $\Pi$ implies the recognizability of $\Pi'$.

\begin{claim}\label{cl:eq-reco}
If $\Pi$ is $k$-recognizable as a property of graphs, then $\Pi'$ is recognizable as a property of rooted forests labelled with $\Sigma$.
\end{claim}
\begin{proof}
Observe that all the malformed trees over $\Sigma$ are $\Pi'$-equivalent. 
Moreover, suppose $s_1,s_2$ are two well-formed trees whose decodings use the same set of interface names. 
Recall that this set of names is some subset $I$ of $\set{1,\ldots,k+1}$ of size at most $k$.
Suppose further that, after renaming $I$ to a subset of $\set{1,\ldots,k}$ in both graphs in the same manner, the decodings turn out to be $\Pi$-equivalent as $k$-interface graphs. 
Then it follows directly from the definitions that $s_1$ and $s_2$ are $\Pi'$-equivalent.
As the number of different sets of interface names over $\set{1,\ldots,k+1}$ is finite, this implies that if $\Pi$-equivalence has finitely many classes of abstraction, then so does $\Pi'$-equivalence.
\cqed\end{proof}

Now, we can use the result of Courcelle~\cite[Theorem 5.3]{Courcelle90} which says that recognizability of properties of forests is equivalent to their definability in counting \mso. 
More precisely, from~\cite{Courcelle90} the following claim is immediate.
\begin{claim}[\cite{Courcelle90}]\label{cl:eq-forests}
If a property $\Pi$ of rooted forests over some finite alphabet is recognizable, then there is a formula $\varphi_\Pi$ of counting \mso that is true in exactly those forests that satisfy $\Pi$.
\end{claim}
We remark that Courcelle~\cite{Courcelle90} works over edge-labelled trees instead of node-labelled forests, but we can easily reduce our setting to his as follows: 
add an artificial root with all the previous roots attached as children, and then push all labels from nodes to the edges connecting them with their parents.
Also, Courcelle~\cite{Courcelle90} defines recognizability by the existence of a homomorphism into a finite algebra, which is easily equivalent to our definition via a Myhill-Nerode equivalence relation.

We conclude the proof by applying Claim~\ref{cl:eq-forests} to property $\Pi'$, which is recognizable due to Claim~\ref{cl:eq-reco}, and
then using the Backwards Translation Theorem to translate the obtained formula $\varphi_{\Pi'}$ back through the composition of transductions $\cal C$ and $\cal L$.
\end{proof}

\section{Omitted proofs from Section~\ref{sec:guidance}}\label{app:guidance}

In the proofs from this section, as well as some proofs in other parts of the appendix, we will use the following notation borrowed from Grohe and Marx~\cite{GroheM15}. 
Suppose $t$ is a tree decomposition of a graph $G$, and $x$ is some its node. Then we use the following notation:
\begin{itemize}
\item $\bag(x)$ is the bag of $x$;
\item $\sep(x)$ is the adhesion of $x$;
\item $\cone(x)$ is the cone of $x$;
\item $\comp(x)$ is the component of $x$;
\item $\mrg(x)$ is the margin of $x$.
\end{itemize}
It will be always clear from the context, to which tree decomposition $t$ these operators are referring.

\subsection{Shorter proofs: Lemmas \ref{lem:sanitation} and~\ref{lem:hyper-pilipczuk}}

\begin{proof}[Proof of Lemma~\ref{lem:sanitation}]
In the following, for a tree decomposition $s$, by $V(s)$ and $E(s)$ we denote the node set and the edge set of $s$, respectively. 
We define the {\em{potential}} of a tree decomposition $s$ as follows: $$\Phi(s)=\sum_{x\in V(s)} |\comp(x)|^2.$$ 
Let $s$ be a tree decomposition of the input graph $G$ that is chosen as follows:
\begin{enumerate}[(i)]
\item\label{p:subset} Each bag of $s$ is a subset of some bag of $t$.
\item\label{p:potential} Among decompositions satisfying~\eqref{p:subset}, $s$ minimizes the potential $\Phi(s)$.
\item\label{p:tot-size}  Among decompositions satisfying~\eqref{p:subset} and~\eqref{p:potential}, $s$ has the minimum total size measured as $$|V(s)|+|E(s)|+\sum_{x\in V(s)} |\bag(x)|.$$ 
\end{enumerate}
Decomposition $s$ is well defined due to $t$ satisfying property~\eqref{p:subset}. We now verify that $s$ is sane.

First, suppose that condition~\eqref{p:notcontained} is not satisfied for some node $x$. 
If $x$ has some parent $y$, then contracting edge $xy$ in $s$ and placing $\bag(y)$ at the node resulting from the contraction yields a tree decomposition $s'$ of $G$.
This tree decomposition has either strictly lower potential $\Phi(\cdot)$ (in case $\comp(x)\neq \emptyset$), or the same potential $\Phi(\cdot)$ but strictly smaller size. 
As the bags of $s'$ are a subfamily of the bags of $s$, this contradicts the choice of $s$. 
Also, if $x$ is the root in $s$ and $\mrg(x)=\emptyset$, then in fact $\bag(x)=\emptyset$. 
Then we can remove $x$ from the decomposition and make its children roots of respective subtrees. 
This yields a tree decomposition of $G$ with the same bags, not larger potential $\Phi(\cdot)$, and strictly smaller size.

For condition~\eqref{p:connectivity}, we only verify the connectivity of $G[\comp(x)]$, because the connectivity of $G[\cone(x)]$ follows from the connectivity of $G[\comp(x)]$ and condition~\eqref{p:neighbors}. 
Suppose then that $G[\comp(x)]$ is not connected for some node $x$. 
This means that $\comp(x)$ can be partitioned into nonempty sets $A$ and $B$ with no edge between them in $G$. Obtain a tree decomposition $s'$ of $G$ as follows. 
Let $p$ be the subtree of $s$ rooted at $x$. 
Replace $p$ in $s$ by two copies $p_A$ and $p_B$ of $p$, where in $p_A$ every bag is replaced with its intersection with $A\cup \sep(x)$, and likewise for $p_B$. 
Since there is no edge in $G$ between $A$ and $B$, it is easy to verify that $s'$ obtained in this manner is indeed a tree decomposition of $G$. 
Moreover, if $y$ is a node of $p$, then the components of the copies of $y$ in $p_A$ and in $p_B$ form a partition of the component of $y$ in $p$.
By the convexity of function $a\mapsto a^2$ we infer that $\Phi(s')$ is not larger than $\Phi(s)$. 
Moreover, the equality can hold only if for every node $y$ of $p$, the component of $y$ is being partitioned trivially: it is either entirely contained in $A$ or entirely contained in $B$. 
However, this is not the case for $y=x$, as both $A$ and $B$ are non-empty.
This implies that in fact $\Phi(s')<\Phi(s)$. 
As every bag of $s'$ is contained in a bag of $s$, which in turn is contained in a bag of $t$, this is a contradiction with the choice of $s$.

Finally, suppose that condition~\eqref{p:neighbors} is not satisfied for some node $x$ and some vertex $u$ of $\sep(x)$. 
Since $\sep(x)\neq \emptyset$, $x$ has some parent $y$ and $u\in \bag(y)$. 
Obtain decomposition $s'$ by removing $u$ from every bag of every descendant of $x$, including $x$ itself. 
Since $u$ is still contained in $\bag(y)$, which in turn contains whole $\sep(x)$, and $u$ has no neighbors in $\comp(x)$, it follows that $s'$ is still a tree decomposition of $G$.
Moreover, it is easy to see that during the construction of $s'$, the component of each node could have only shrunk; hence, $\Phi(s')\leq \Phi(s)$.
Also $s'$ is obtained from $s$ by removing some vertices from some bags, so in particular the total size of $s'$ is strictly smaller than that of $s$. 
This is a contradiction with the choice of $s$.
\end{proof}

\begin{proof}[Proof of Lemma~\ref{lem:hyper-pilipczuk}]
\emph{Ad~\eqref{cnd:hyp-pil:left}.} 
Let $\Lambda$ be a guidance system over $G$ that captures $\Xx-u$. 
Consider the connected component of $u$ inside $G$. 
Choose a spanning tree of this component, and direct its edges toward $u$ so that it becomes the root; call the resulting tree $t$. 
We claim that the guidance system $\Lambda \cup \set t$ captures $\Xx$, thus proving~\eqref{cnd:hyp-pil:left}, because at most one more color is needed to color the added tree $t$. 
Consider a set $X \in \Xx$. 
If $X$ is $\set{u}$, then it is captured by $u$ using the added tree $t$. 
Otherwise, $X-u$ is a nonempty set that is contained in $\Lambda(v)$ for some  $v$. If $X$ does not contain $u$, then we are done. 
Otherwise, $X$ is contained in the connected component of $u$, and therefore the added tree $t$ takes care of $u$.

\smallskip

\noindent\emph{Ad~\eqref{cnd:hyp-pil:right}.} Let $\Lambda$ be a guidance system in $G$ that captures $\Xx$. Let
\begin{align*}
	t_1, \ldots, t_n
\end{align*}
be the trees from $\Lambda$ that contain vertex $u$. By assumption on $\Lambda$ being $k$-colorable, there are at most $k$ such trees.  
For each $i \in \set{1,\ldots,n}$ choose a spanning tree $s_i$ of the component in $G-u$ that contains the root of $t_i$, and direct all edges toward this root.
Whenever this root is equal to $u$, by slightly abusing the notation take $s_i$ to be an empty tree; we will note use them in the final guidance system.
Define a new guidance system
\begin{align*}
	\Lambda'=\Lambda - \set{t_1,\ldots,t_n} \cup \set{s_1,\ldots,s_n}.
\end{align*}
This is a guidance system over $G-u$, and at most $n\leq k$ additional colors are needed to color trees $s_1,\ldots,s_n$.
In order to verify that $\Lambda'$ captures $\Xx$, take any set $X$ of $\Xx$.
Then $X$ is contained in some component $C$ of $G-u$, and is captured by some vertex $v$ in~$\Lambda$.

Suppose first that $v$ belongs to $C$.
We claim that then $X$ is also captured by $v$ in $\Lambda'$.
Indeed, for every element $x$ of $X$, either the tree that certifies that $x\in \Lambda(v)$ survives in $\Lambda'$, or it is one of $\set{t_1,\ldots,t_n}$. 
In the latter case, the corresponding tree $s_i$ certifies that $x\in \Lambda'(v)$, due to $v$ and $x$ being in the same connected component of $G-u$ and the way $s_i$ is defined.

Suppose now that $v$ belongs to some other component of $G-u$.
Then, every tree of $\Lambda$ that simultaneously contains $v$ and some element of $X$, must also contain $u$.
Consequently, every tree that certifies that $X$ is captured by $v$ has to be among $\set{t_1,\ldots,t_n}$.
By the construction of trees $s_i$, it follows that $X$ is captured in $\Lambda'$ by any vertex of component $C$.
\end{proof}

\subsection{Proof of Lemma~\ref{lem:capture-adhesion-capture-tree-decomposition}}\label{app:guidance-long-proof}

This section is entirely devoted to the proof of Lemma~\ref{lem:capture-adhesion-capture-tree-decomposition}. 
We first show how guidance systems can be guessed in \mso, and then apply this understanding to prove the lemma.

\paragraph*{Describing a guidance system in \mso.}
\newcommand{\adhegraph}{G\langle t \rangle}
\newcommand{\compset}{\mathsf{comp}}
\newcommand{\childset}{\mathsf{cparent}}
\newcommand{\nodefun}{\mathsf{node}}
\newcommand{\graphsets}{\decograph G \Xx}
\newcommand{\belongs}{\mathsf{element}}
\newcommand{\decograph}[2]{{#1}^{#2}}

The point of guidance systems is that they can be encoded by an \mso transduction, as stated in Lemma~\ref{lem:code-guidance-in-mso} below. 
We use the name \emph{$k$-decorated graph} for a graph $G$ with a distinguished family $\Xx$ of subsets of vertices, where each subset from $\Xx$ has size at most $k$.
A $k$-decorated graph is represented as a logical  structure $\decograph G \Xx$ as follows. 
The universe is the vertices and edges of $G$, plus a new element for each set from $\Xx$. 
The vocabulary consists of two binary relations $\inc,\belongs$ and one unary relation $\decoration$, where
$\inc$ describes the incidence between vertices and edges in $G$, $\decoration$ selects elements representing sets $X\in \Xx$, and $\belongs$ selects pairs $(v,X)$ with $v \in X \in \Xx$.

We also define the {\em{enriched}} representation of a $k$-decorated graph $\decograph G \Xx$ as follows.
In addition to the predicates defined above, the structure contains also predicate $\belongs_i$ for each $i\in \set{1,2,\ldots,k}$.
We require the following two conditions:
\begin{itemize}
\item Relation $\belongs$ is the disjoint union of all relations $\belongs_i$; that is, $\belongs(u,X)$ holds if and only if $\belongs_i(u,X)$ holds for some $i$, 
      and $\belongs_i(u,X)$ and $\belongs_j(u,X)$ cannot hold simultaneously for $i\neq j$.
\item For each set $X\in \Xx$ and each $i\in \set{1,\ldots,k}$, there are no two different vertices $u,u'$ for which $\belongs_i(u,X)$ and $\belongs_i(u',X)$ hold simultaneously.
\end{itemize}
Intuitively, the enriched representation is also supplied with some coloring of relation $\belongs$, using which every set $X\in \Xx$ can distinguish between its elements using indices from $1$ to $k$.
Note that every $\decograph G \Xx$ has a unique representation, but multiple enriched representations: the relation $\belongs$ may be partitioned in different ways into relations $\belongs_i$ to
satisfy the conditions above.

\newcommand{\rt}{\mathsf{root}}

\begin{lemma}\label{lem:code-guidance-in-mso}
Let $k \in \Nats$. There exists an \mso transduction from graphs to enriched representations of decorated graphs with the following property:
when given a graph $G$ on the input, the transduction outputs some enriched representation of every $k$-decorated graph $\decograph G \Xx$ 
for which $\Xx$ can be captured by some $k$-colorable guidance system in~$G$.
\end{lemma}
\begin{proof}
We construct an \mso transduction by composing the following seven steps. The \mso-definability of the stated properties will be always straightforward.
\begin{enumerate}[(1)]
\item\label{st:1} Using coloring, guess $3k$ subsets 
\begin{align*}
F_1,\ldots,F_k\quad\textrm{and}\quad X_1,\ldots,X_k \quad\textrm{and}\quad R_1,\ldots,R_k.
\end{align*}
\item\label{st:2} Using interpretation that uses only the domain check $\domain$ and preserves the structure intact, verify the following for each $i=1,2,\ldots,k$:
\begin{itemize}
\item Sets $X_i$ and $R_i$ are subsets of vertices, and $R_i\subseteq X_i$.
\item Set $F_i$ is an acyclic subset of edges, and every edge of $F_i$ has both endpoints in $X_i$.
\item Each connected component of the forest $(X_i,F_i)$ (i.e. forest with vertices $X_i$ and edges $F_i$) contains exactly one vertex from $R_i$.
\end{itemize}
\item\label{st:3} Using interpretation, for each $i=1,\ldots,k$ introduce a predicate $$\rt_i(u,v)$$ with the following semantics: $\rt_i(u,v)$ is true if, and only if $u$ belongs to $X_i$ and $v$ is the
unique vertex of $R_i$ which belongs to the connected component of the forest $(X_i,F_i)$ that contains~$u$.
\item\label{st:4} Copy the universe $1+2^k$ times. Let us index the layers of the copied universe as
$$L_0\cup \set{L_A\colon A\subseteq \set{1,\ldots,k}}.$$
Let $L'$ be equal to the union of all layers apart from~$L_0$. 
\item\label{st:5} Using interpretation, define predicate $\belongs_i(v,e)$ as follows: 
$\belongs_i(v,e)$ is true if, and only if $v$ is a vertex from $L_0$ and $e$ is a copy of a vertex contained in some $L_A$, for which $i\in A$ and $\rt_i(e,v)$ holds.
Define $\belongs$ as the union of relations $\belongs_i$.
\item\label{st:6} Using coloring and interpretation, guess any unary predicate $\decoration$ that is true only in some elements of $L'$ that are copies of vertices. 
\item\label{st:7} Using interpretation again, clean the structure using the following steps:
\begin{itemize}
\item Trim the universe to the elements belonging to $L_0$ or satisfying $\decoration$. 
\item Restrict all the predicates to $L_0$, with the exception of $\belongs$, $\belongs_i$, and $\decoration$. 
\item Remove predicates $\rt_i$.
\item Verify in the domain check $\domain$ that the output is an enriched representation of a $k$-decorated graph.
\end{itemize}
\end{enumerate}
Clearly, the composition of the steps above is an \mso transduction. 
From the construction it should be clear that for every $k$-decorated graph $\decograph G \Xx$, for which $\Xx$ can be captured by some $k$-colorable guidance system in $G$, 
some enriched representation of $\decograph G \Xx$ will be output by the interpretation.
Indeed, in steps \eqref{st:1}--\eqref{st:3} we guess a $k$-colorable guidance system $\guid$ by guessing the forest corresponding to each color separately, 
and we encode the function $\guid(\cdot)$ using predicates $\rt_i$. In steps \eqref{st:4}--\eqref{st:7} we guess and construct the decorated graph $\decograph G \Xx$: each decoration $X\in \Xx$ is encoded
using a vertex of the graph which captures $X$ in $\guid$, plus the set of colors of trees that lead from this vertex to the vertices of~$X$. 
The partition of relation $\belongs$ into relations $\belongs_i$ is defined according to the colors of trees in the guessed guidance system.
\end{proof}

In order to prove Lemma~\ref{lem:capture-adhesion-capture-tree-decomposition}, it suffices to combine Lemma~\ref{lem:code-guidance-in-mso} with the following lemma, 
which says that a transduction can recover a sane tree decomposition from its family of adhesions, assuming that the adhesions are not larger than $k$.

\begin{lemma}\label{lem:from-decorated-to-sane}
Let $k \in \Nats$. 
There is an \mso transduction from enriched representations of $k$-decorated graphs to tree decompositions, 
which maps an enriched representation of a $k$-decorated graph $\decograph G \Xx$ to all sane tree decompositions of $G$ for which the family of adhesions is $\Xx$. 
\end{lemma}
\begin{proof}
We first explain how any sane tree decomposition $t$ can be encoded using a very simple object in a graph, essentially a forest colored with two colors.
For this, we first need some simple claims about sane tree decompositions.

\begin{claim}\label{cl:exists-direct}
Suppose $t$ is a sane tree decomposition of a graph $G$, $y$ is a node of $t$, and $x$ is the parent of $y$.  Then $$\sep(y)\cap \mrg(x)\neq \emptyset,$$
that is, there exists a vertex of $G$ that is simultaneously in the adhesion of $y$ and in the margin of its parent $x$.
\end{claim}
\begin{proof}
By condition~\eqref{p:notcontained} of Definition~\ref{def:sane}, both margins $\mrg(x)$ and $\mrg(y)$ are not empty.
They are also disjoint and contained in $\comp(x)$, which, by condition~\eqref{p:connectivity} of Definition~\ref{def:sane}, is connected in $G$.
Hence, in $G$ there is a path from a vertex of $\mrg(y)$ to a vertex of $\mrg(x)$ that traverses only vertices contained within $\comp(x)$.
Since $\mrg(y)$ is a subset of $\comp(y)$ and $\mrg(x)$ is disjoint with $\comp(y)$, this path needs to traverse some vertex $u$ of $\sep(y)$. 
Since the path is entirely contained in $\comp(x)$, we have that $u\in \sep(y)\cap \comp(x)=\sep(y)\cap \mrg(x)$.
\cqed\end{proof}

Suppose we are given a sane tree decomposition $t$ of a graph $G$. 
Define $\adhegraph$ to be the graph obtained from $G$ by adding an edge between every pair of nonadjacent vertices of $G$ that appear together within some adhesion of $t$;
thus, $\adhegraph$ is obtained from $G$ by turning every adhesion into a clique.
Obviously, $t$ is still a sane tree decomposition of $\adhegraph$.
Note that the subgraphs of $\adhegraph$ induced by the margins of the nodes of $t$ are exactly the marginal graphs of these nodes. 
Hence, they are nonempty and connected due to the saneness of $t$.
We now show that the marginal graphs also have connections to the adhesion of the corresponding nodes.

\begin{claim}\label{lem:exists-neighbor}
Suppose $t$ is a sane tree decomposition of a graph $G$, and $x$ is some its node. Then every vertex of $\sep(x)$ has a neighbor in $\mrg(x)$ in the graph $\adhegraph$.
\end{claim}
\begin{proof}
Fix a vertex $v\in \sep(x)$. Take any vertex $u\in \mrg(x)$. 
By the saneness of $t$, there is a path $P$ in $G$ that goes from $u$ to $v$ and which, apart from the endpoint $v$, is entirely contained in $\comp(x)$.
Let $u'$ be the last vertex on this path that belongs to $\mrg(x)$; then all the internal vertices on the suffix of $P$ between $u'$ and $v$ need to be outside of $\bag(x)$.
If there are no such internal vertices, i.e. this suffix consists only of $u'$ and $v$, then we are done because $u'$ belongs to $\mrg(x)$ and is a neighbor of $v$ in $G$.
Otherwise, all these internal vertices must be contained in the component of some child $y$ of $x$, and in particular both $u'$ and $v$ belong to the adhesion of $y$.
Then from the construction of $\adhegraph$ it follows that $v$ and $u'$ are adjacent in this graph.
\cqed\end{proof}

Claims~\ref{cl:exists-direct} and~\ref{lem:exists-neighbor} suggest the following representation of a sane tree decomposition $t$. The representation consists of sets $M,K,R$, where
\begin{itemize}
\item $M$ is a subset of edges of $\adhegraph$ obtained as follows: 
for each node $x$ of $t$, take an arbitrary spanning tree of the marginal graph of $x$, and as $M$ take the union of the edge sets of these spanning trees.
\item $K$ is a subset of edges of $\adhegraph$ obtained as follows:
for each node $y$ of $t$ with parent $x$, take any vertex $v\in \sep(y)\cap \mrg(x)$ (which exists by Claim~\ref{cl:exists-direct}), 
and add to $K$ an arbitrary edge of $\adhegraph$ connecting it to some vertex of $\mrg(y)$ (which exists by Claim~\ref{lem:exists-neighbor}).
\item $R$ is a subset of vertices of $G$ that has exactly one vertex in every margin of a root of $t$, and no other vertices.
\end{itemize}
Any such triple $(M,K,R)$ will be called an {\em{encoding}} of $t$.
It is not hard to see that, given graph $\adhegraph$ together with a triple $(M,K,R)$ that is an encoding of $t$, 
one can uniquely decode the whole decomposition $t$. We do it in the following claim using an \mso transduction.

\begin{claim}\label{cl:from-adhe-to-end}
There is an \mso transduction, whose domain are graphs supplied with a triple  of universe subsets, that has the following property:
if the input is $\adhegraph$ together with an encoding $(M,K,R)$ of $t$, for some sane tree decomposition $t$ of a graph $G$, then the output of the transduction contains $t$ with $\adhegraph$ as the underlying graph.
\end{claim}
\begin{proof}
The \mso transduction takes the following steps (all of them are trivially definable in \mso):
\begin{enumerate}[(1)]
\item Using an interpretation that keeps the structure intact and just makes some domain check, verify the following:
\begin{itemize}
\item $M$ and $K$ are disjoint subsets of edges of $\adhegraph$, 
\item $M\cup K$ is a spanning forest of $\adhegraph$, and 
\item $R$ contains exactly one vertex in each connected component of $\adhegraph$.
\end{itemize}
\item Using coloring, copying, and interpretation, guess one vertex $u_C$ from each connected component $C$ of the subgraph spanned by $M$, and create its copy $x_C$. 
This copy will be used as the node in the output tree decomposition whose margin is the vertex set of $C$.
\item For each node $x_C$, define the parent of $x_C$ to be the node $x_{C'}$ with the following property:
If $P$ is the unique path in $M\cup K$ leading from $u_C$ to a vertex of $R$, then $C'$ is the next connected component of the subgraph spanned by $M$ that is visited by $P$.
From now on we can talk about ancestor-descendant relation between nodes $x_C$. The goal now is, having the tree structure and the set of margins, to decode the bags of $t$
\item For each node $x_C$, define the component of $x_C$ to be the union of the vertex sets of all $C'$ for which either $x_{C'}=x_C$, or $x_{C'}$ is a descendant of $x_C$.
\item Finally, for each node $x_C$ define the bag of $x_C$ as the vertex set of $C$, plus all the neighbors of vertices in the component of $x$ that do not belong to this component.
\item At the end, clean the structure from unnecessary relations and verify that the output is a sane tree decomposition whose encoding is $(M,K,R)$.
\end{enumerate}
It is straightforward to see that if $(M,K,R)$ is an encoding of a sane tree decomposition $t$, then $t$ will be the unique (up to isomorphism) decomposition output by the transduction.
\cqed\end{proof}

Therefore, it remains to argue how to construct a suitable graph $\adhegraph$ from the input graph $G$, because then we will be able to guess sets $M,K,R$ using coloring,
run the transduction of Claim~\ref{cl:from-adhe-to-end}, drop the added edges from the underlying graph, and at the end verify that the output is a sane tree decomposition of $G$ whose family of adhesions is $\Xx$.
This is, however, very easy.
Namely, for all sane tree decompositions $t$ for which $\Xx$ is the family of adhesions, the graphs $\adhegraph$ are isomorphic: they are obtained from $G$ by making each set of $\Xx$ into a clique.
Since all the sets of $\Xx$ are of size at most $k$ and we are working with an enriched representation, the construction of this graph can be easily done using an \mso transduction as follows. 
Create $\binom{k}{2}$ copies of each decoration $X\in \Xx$, and index them using 2-element subsets of $\set{1,2,\ldots,k}$.
For a copy corresponding to a pair $\{i,j\}\subseteq \set{1,\ldots,k}$, verify whether there is a pair of vertices $u,v$ in $X$ for which $\belongs_i(u,X)$ and $\belongs_j(u,X)$ hold, 
and moreover $u$ and $v$ are nonadjacent in $G$. By the definition of an enriched representation, there is at most one such pair $u,v$.
If the pair $u,v$ does not exist, then remove the copy from the universe, and otherwise turn the copy into an edge between $u$ and $v$.
Observe that, in this manner, an edge $uv$ could have been added multiple times to the structure, in case $\set{u,v}$ is a subset of multiple sets from $\Xx$. 
However, we can remove the multiplicities by guessing one duplicate of each edge to preserve,
and removing all the others.
\end{proof}

\section{Omitted proofs from Section~\ref{sec:pathwidth}}\label{app:pathwidth}

\begin{proof}[Proof of Lemma~\ref{lem:compute-path-decomposition}, assuming Lemma~\ref{lem:pathwidth-comb}]
Let $f$ be the function from Lemma~\ref{lem:pathwidth-comb}, and let $k \in \Nats$. 
Applying Lemma~\ref{lem:capture-adhesion-capture-tree-decomposition}, there is an \mso transduction $\Ss_{f(k)}$ that, given a graph, outputs all its sane tree decompositions for which
the family of adhesions can be captured by an $f(k)$-colorable guidance system. 
Thanks to Lemma~\ref{lem:pathwidth-comb}, we know that if the input graph has pathwidth at most $k$, then some its tree decomposition $t$ is captured by an $f(k)$-colorable guidance system.
By applying Lemma~\ref{lem:sanitation}, we can further assume that $t$ is sane; indeed, if a guidance system captures some bag, then it also captures all its subsets.
Since adhesions of $t$ are subsets of bags of $t$, we infer that the same guidance system also captures the family of adhesions of $t$, and hence $t$ is among the outputs of $\Ss_{f(k)}$.
To conclude the proof it remains to take $\Ss_{f(k)}$ and add a verification (in the domain check $\domain$) that the output decomposition has width at most $f(k)$, which clearly can be expressed in \mso
for a constant~$k$.
\end{proof}

\begin{proof}[Proof of Lemma~\ref{lem:basic-properties-of-gtw}]
\emph{Ad~\eqref{eq:basic-dunion}.} We prove inequalities in both directions. 
First, if $t$ is a tree decomposition of $G\dunion G'$ captured by some guidance system $\Lambda$, then restricting $\Lambda$ only to trees contained in $G$ yields a guidance system 
that captures a tree decomposition of $G$ --- the one derived from $t$ by intersecting every bag with the vertex set of $G$. 
By performing the same reasoning for $G'$, we conclude that $\gtw(G\dunion G')\geq\max(\gtw(G),\gtw(G'))$.
Second, if $t$ and $t'$ are tree decompositions of $G$ and $G'$, respectively, then their disjoint union is a tree decomposition of $G \dunion G'$. 
This is because in Definition~\ref{def:tree-decomposition}, we allow a tree decomposition to be a forest. 
If $t$ and $t'$ are captured by $k$-colorable guidance systems, over disjoint graphs, then the union of these guidance systems is also $k$-colorable and captures $t\dunion t'$.
This proves that $\gtw(G\dunion G')\leq\max(\gtw(G),\gtw(G'))$.

\smallskip

\noindent\emph{Ad~\eqref{eq:basic-left-rem}.} Using~\eqref{eq:basic-dunion}, we assume that $G$ is connected. 
Let $t$ be a tree decomposition of $G-u$ which can be captured by a $k$-colorable guidance system over $G-u$, and therefore also over $G$. 
Add the vertex $u$ to every bag, creating a new decomposition. 
Since $G$ is connected, every bag of the new tree decomposition is contained in some connected component of $G$. 
By claim~\eqref{cnd:hyp-pil:left} of Lemma~\ref{lem:hyper-pilipczuk}, every bag of the new decomposition can be captured by a guidance system with at most $k+1$ colors over $G$.

\smallskip

\noindent\emph{Ad~\eqref{eq:basic-right-rem}.} Let $t$ be a tree decomposition of $G$ which is captured by a $k$-colorable guidance system over $G$.
Define a tree decomposition of $G-u$ as follows: for each connected component of $G-u$, create a tree decomposition by restricting all bags of $t$ to that component, 
and then take the union of all of these tree decompositions. 
Every bag of the new tree decomposition is  contained in the intersection of some bag of $t$ with some connected component of $G-u$. 
Therefore, by claim~\eqref{cnd:hyp-pil:right} of Lemma~\ref{lem:hyper-pilipczuk}, the new tree decomposition can be captured by a $2k$-colorable guidance system over $G-u$.\end{proof}

\section{Omitted proofs from Section~\ref{sec:treewidth}}

\begin{proof}[Proof of Lemma~\ref{lem:order}]
Take any path $P$ from the source to the sink, and let $e_1,e_2,\ldots,e_p$ be the order in which the cutedges appear on this path.
Suppose there is a path $P'$ for which the $i$-th cutedge appearing on $P'$ is different than $e_i$, and let $i$ be the smallest index with this property for $P'$.
Then the $i$-th cutedge appearing on $P'$ is $e_j$ for some $j>i$.
Construct a path $Q$, leading from the source to the sink, by first traversing a prefix of $P'$ up to the cutedge $e_j$, and then traversing a suffix of $P$ from the cutedge $e_j$. 
Since neither this prefix nor this suffix traverses $e_i$, we conclude that $Q$ is a path from the source to the sink that avoids $e_i$.
This is a contradiction with $e_i$ being a cutedge.
\end{proof}

\begin{proof}[Proof of Lemma~\ref{lem:incidence-cutedges}]
Let $H$ be the underlying hypergraph of the network. By a {\em{source-sink path}} we will mean a path going from the source of the network to its sink.
By the connectedness of $H$, every cutedge component intersects at least one cutedge, i.e., one of elements $e_1,\ldots,e_p$.

Let $C$ be a cutedge component of $H$. 
For the sake of contradiction, suppose that $C$ intersects some two elements that are not consecutive in the sequence; say $e_i$ and $e_j$, where there exists some $\ell$ with $i<\ell<j$.
Observe that excluding this situation already gives us all the claims stated in the lemma.

We first consider the case when $1\leq i<j\leq p$, i.e., $e_i$ and $e_j$ are cutedges of the network, and not the singletons of the source or sink.
Take any source-sink path $P$; by Lemma~\ref{lem:order} we know that $P$ visits the cutedges in the order given by the cutedge sequence.
Construct a source-sink path as follows: take the prefix of $P$ from the source to cutedge $e_i$, 
then travel inside the component $C$ from a vertex belonging to $e_i$ to a vertex belonging to $e_j$, and 
finally follow the suffix of $P$ from the cutedge $e_j$ to the sink. 
Thus, the constructed path leads from the source to the sink and avoids $e_\ell$.
This is a contradiction with $e_\ell$ being a cutedge.

By symmetry, we are left with considering the case when $i=0$, i.e., $e_i$ is the singleton of the source.
We can perform almost the same construction as before:
If $P$ is any source-sink, then we can obtain a source-sink path that avoids $e_\ell$ by concatenating 
any path that leads from the source to a vertex of $e_j$ within component $C$, and the suffix of $P$ from $e_j$ to the sink.
Again, this contradicts the fact that $e_\ell$ is a cutedge.
\end{proof}

\begin{proof}[Proof of Lemma~\ref{lem:menger-str}]
Let $(e_1,\ldots,e_p)$ be the cutedge sequence of the network.
Let $H'$ be the hypergraph of the network with all the cutedges removed; the connected components of $H'$ are bridges and appendices.
Let $V_i$ be the union of $(e_i,e_{i+1})$-bridges, for $i=0,1,\ldots,p$, and let $W_i$ be the union of $e_i$-appendices, for $i=1,2,\ldots,p$.
For any $i=0,1,\ldots,p$, consider the hypergraph $H'[V_i]$. 
By applying Menger's theorem in its incidence graph, i.e., the bipartite graph formed by its vertices and hyperedges, where edge relation is defined by containment, 
we infer that in $H'[V_i]$ there are two hyperedge-disjoint paths $P^1_i$ and $P^2_i$ leading from $e_i\cap V_i$ to $e_{i+1}\cap V_i$.
Indeed, otherwise $H'[V_i]$ would contain a hyperedge whose removal would disconnect $e_i\cap V_i$ from $e_{i+1}\cap V_i$; 
this hyperedge would be also a cutedge of the network that would not be among $e_1,e_2,\ldots,e_p$. 
Now construct $P^1$ by concatenating alternately paths $P^1_0,P^1_1,\ldots,P^1_p$ and cutedges $e_1,e_2,\ldots,e_p$, and similarly construct $P^2$. 
It follows that both $P^1$ and $P^2$ lead from the source to the sink, and the only hyperedges shared by them are $e_1,e_2,\ldots,e_p$. 
\end{proof}

\begin{proof}[Proof of Lemma~\ref{lem:thin-pw}]
We adopt the notation of sets $V_i$ and $W_i$ from the definition of $k$-thinness. 
For $i=0,1,\ldots,p$, let $t_i$ be a path decomposition of $H[V_i]$ that certifies $k$-thinness.
Similarly, let $s_i$ be a path decomposition of $H[W_i]$ that certifies $k$-thinness. 
Construct $s'_i$ from $s_i$ by adding all the vertices of $e_i$ to all the bags.
As $|e_i|\leq k$, the width of $s'_i$ does not exceed $2k$.
 
Construct a path decomposition of the given network by concatenating decompositions $$t_0,s'_1,t_1,s'_2,t_2,\ldots,t_{p-1},s'_p,t_p$$ in this order.
It is easy to verify that this is indeed a path decomposition of the given network, and its width is at most $2k+1$.
\end{proof}

\newcommand{\Hff}{\widehat{\Hf}}
\newcommand{\Ht}{\widehat{H}}

\begin{proof}[Proof of Lemma~\ref{lem:thin-persist}]
Let $\Hf$ be the given network, let $u$ and $v$ be the source and the sink of $\Hf$, and let $H$ be its hypergraph. 
Let $(e_1,\ldots,e_p)$ be the cutedge sequence of $\Hf$, and let us adopt the notation for sets $V_i$ and $W_i$ from the definition of $k$-thinness.
As $e$ is the cutedge of $\Hf$, we have that $e=e_\ell$ for some $\ell\in \set{1,\ldots,p}$.

Let $\Ht=H[e\to K]$. Since $K$ is connected, by the definition of $\Ht$ it follows that $\Ht$ is connected as well. 
Hence, $\Ht$ with source $u$ and sink $v$ is a network, which we shall call $\Hff$.
In the following, by a {\em{source-sink}} path we mean a path going from $u$ to $v$.
Our first goal is to understand the structure of the cutedges of $\Hff$.

First, note that each $e_i$ with $i\neq \ell$ is still a cutedge in $\Hff$. 
Moreover, these cutedges must be ordered in the same manner in the cutedge sequence of $\Hff$ as in the cutedge sequence of $\Hf$.
This is because any source-sink path in $\Hf$ can be lifted to a source-sink path in $\Hff$ by replacing the usage of $e_\ell$ with a path within $K$. 
Then the order in which the cutedges of $\Hf$ appear in the lifted path is the same as in the original one.
The same argument shows that in the cutedge sequence of $\Hff$, all cutedges that originate from the hypergraph $K$ have to be after all cutedges $e_j$ for $j<\ell$, 
and before all cutedges $e_j$ for $j>\ell$.

Next, we observe that every hyperedge $f$ of $H$ that is not a cutedge in $\Hf$, does not becomes a cutedge in $\Hff$. 
Indeed, Lemma~\ref{lem:menger-str} implies that in $\Hf$ there is a source-sink path that avoids $f$.
By lifting this path in the same manner as in the previous paragraph, we obtain a source-sink path in $\Hff$ that avoids $f$; this certifies that $f$ is not a cutedge in $\Hff$.

We conclude that the cutedge sequence of $\Hff$ has the following form: 
$$(e_1,\ldots,e_{\ell-1},f_1,\ldots,f_q,e_{\ell+1},\ldots,e_p),$$
where $f_1,f_2,\ldots,f_q$ are the cutedges of $\Hff$ that originate in $K$.
In what follows we assume for simplicity that $q>0$; at the end we shortly discuss how the proof needs to be adjusted in the case when no new cutedge arises in the network.

\begin{figure}[htbp!]
        \centering
                \def\svgwidth{\columnwidth}
                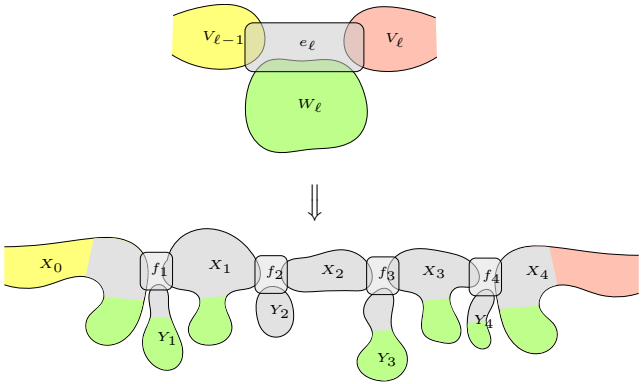
\caption{Part of the cutedge decomposition affected by replacement of a cutedge $e$; here, $q=4$.}\label{fig:replacement}
\end{figure}

From Lemma~\ref{lem:incidence-cutedges} it straightforward to see that the bridges and appendices in $\Hff$ have the following structure (see Figure~\ref{fig:replacement}):
\begin{itemize}
\item For each $i\in \set{0,\ldots,p}-\set{\ell-1,\ell}$, every $(e_i,e_{i+1})$-bridge in $\Hf$ remains an $(e_i,e_{i+1})$-bridge in $\Hff$.
\item For each $i\in \set{1,\ldots,p}-\set{\ell}$, every $e_i$-appendix in $\Hf$ remains an $e_i$-appendix in $\Hff$.
\item For each $j\in \set{1,\ldots,q-1}$, every $(f_j,f_{j+1})$-bridge consists only of vertices originating in $K$ or $W_\ell$.
\item For each $j\in \set{1,\ldots,q}$, every $f_j$-appendix consists only of vertices originating in $K$ or $W_\ell$.
\item Each $(e_{\ell-1},f_1)$-bridge consists only of vertices originating in $V_{\ell-1}$, $K$, or $W_\ell$.
\item Each $(f_q,e_{\ell+1})$-bridge consists only of vertices originating in $V_{\ell}$, $K$, or $W_\ell$.
\end{itemize}
In particular, the vertices originating in $K$ and $W_\ell$ are partition between bridges and appendices that intersect cutedges $f_j$. Also, only the $(e_{\ell-1},f_1)$- and $(f_q,e_{\ell+1})$-bridges
may contain some vertices not originating in $K$ or $W_\ell$: each $(e_{\ell-1},e_\ell)$-bridge of the original network $\Hf$ is contained in some $(e_{\ell-1},f_1)$-bridge of $\Hff$, and likewise for
$(e_{\ell},e_{\ell+1})$-bridges of $\Hf$. 

We now verify that $\Hff$ is $k$-thin by exposing appropriate path decompositions for all the relevant parts. Due to the first two items above, 
for the unions of $(e_i,e_{i+1})$-bridges, for $i\in \set{0,\ldots,p}-\set{\ell-1,\ell}$,
and the unions of $e_i$-appendices, for $i\in \set{1,\ldots,p}-\set{\ell}$, we can use the same decompositions as the ones that certify that the original network $\Hf$ is $k$-thin.
Hence we are left only with the relevant unions of bridges and appendices that intersect the new cutedges $f_j$.
For $j=1,\ldots,q+1$, let $X_j$ be the union of the vertex sets of $(f_{j-1},f_j)$-bridges in $\Hff$, where $f_0=e_{\ell-1}$ and $f_{q+1}=e_{\ell+1}$.
Similarly, for $j=1,\ldots,q$, let $Y_j$ be the union of the vertex sets of $f_j$-bridges in $\Hff$.

Consider first the part $X_j$ for some $j\in \set{1,2,\ldots,q-1}$. The vertices of this part originate in $W_\ell\cup V(K)$, where $V(K)$ denotes the vertex set of $K$. 
Recall that $H[W_\ell]$ admits some path decomposition $s$ of width at most $k$ with $W_\ell\cap e_\ell$ contained in the first bag.
If $W_\ell$ is empty, then we take a decomposition consisting of one node with an empty bag; this applies also to similar situations in the next paragraphs.
Consequently, by intersecting each bag of $s$ with $X_j$ we infer that $H[X_j\cap W_\ell]$ admits a path decomposition $s'$ of width at most $k$ with $(W_\ell\cap X_j)\cap e_\ell$ contained in the first bag. 
Recall also that $K$ has at most $k+1$ vertices, so $|X_j\cap V(K)|\leq k+1$. 
Hence, a suitable path decomposition of $\Ht[X_j]$ can be obtained by taking $s'$ and adding $X_j\cap V(K)$ to every its bag; thus, the obtained decomposition has width at most $2k+1$. 
Observe that $X_j\cap f_{j-1}\subseteq X_j\cap V(K)$ and $X_j\cap f_{j-1}\subseteq X_j\cap V(K)$. 
Hence $X_j\cap f_{j-1}$ is contained in the first bag of this path decomposition, whereas $X_j\cap f_{j}$ is contained in the last bag, as required.

Consider now the part $X_0$. 
Since $V_{\ell-1}\subseteq X_0$, we can partition $X_0$ into $V_{\ell-1}$ and $X_0\setminus V_{\ell-1}$, where the latter set will be denoted by $M$.
By the assumption that $\Hf$ is $k$-thin, $H[V_{\ell-1}]$ admits a path decomposition $t$ of width at most $2k+1$ 
where $V_{\ell-1}\cap e_{\ell-1}$ is contained in the first bag and $V_{\ell-1}\cap e_\ell$ is contained in the last bag. 
On the other hand, $M\subseteq W_\ell\cup V(K)$, and hence it can be further partitioned into $M\cap W_\ell$ and $M\setminus W_\ell\subseteq V(K)\cap X_0$. 
As in the previous paragraph, by the assumption of $k$-thinness we infer that $H[M\cap W_\ell]$ admits a path decomposition $s$ of width at most $k$ with $(M\cap W_\ell)\cap e_\ell$ contained in the first bag. 
Obtain a path decomposition $s'$ by adding all the vertices of $V(K)\cap X_0$ to every bag of $s$; since $|V(K)|\leq k+1$, we have that the width of $s'$ does not exceed $2k+1$. 
Finally, obtain a path decomposition of $\Ht[X_0]$ by concatenating decompositions $t$ and $s'$. 
It is easy to see that this is indeed a path decomposition of $\Ht[X_0]$ and its width obviously is at most $2k+1$.
Moreover, $X_0\cap e_{\ell-1}=V_{\ell-1}\cap e_{\ell-1}$ is contained in its first bag (by the assumption about $t$) and $X_0\cap f_1\subseteq V(K)\cap X_0$ is contained in the last bag (by the construction).

A symmetric reasoning can be performed for the part $X_q$.

We are left with considering part $Y_j$, for any $j\in \set{1,\ldots,q}$. 
Again, the vertices of this part originate in $W_\ell \cup V(K)$. 
As before, by the assumption about $k$-thinness of $\Hf$, we have that $H[W_\ell\cap Y_j]$ admits a path decomposition $s$ of width at most $k$ with $(W_\ell\cap Y_j)\cap e_\ell$ contained in the first bag. 
Let $s'$ be a path decomposition obtained from $s$ by prepending (i.e., adding at the front) bag $V(K)\cap Y_j$. 
It is easy to verify that $s'$ is indeed a path decomposition of $\Ht[Y_j]$; 
this follows from the fact that $(V(K)\cap Y_j)\cap (W_\ell\cap Y_j)=(W_\ell\cap Y_j)\cap e_\ell$, and this set is contained in the first bag of $s$. 
Moreover, $s'$ has width at most $k$ due to $|V(K)|\leq k+1$, and its first bag contains $Y_j\cap f_j\subseteq V(K)\cap Y_j$, as required.

Having considered all the parts, we conclude that $\Hff$ is $k$-thin in the case when $q>0$.
To see that case $q=0$ can be analyzed in the same manner, it suffices to combine the arguments for $X_0$ and $X_q$.
Precisely, a path decomposition of $\Ht[X_0]$ is formed by concatenating: 
(i) the path decomposition of $H[V_{\ell-1}]$ certifying $k$-thinness of $\Hf$, 
(ii) the path decomposition of $H[W_\ell]$ certifying $k$-thinness of $\Hf$, with all the vertices of $K$ added to each bag,
and (iii) the path decomposition of $H[V_\ell]$ certifying $k$-thinness of $\Hf$.
\end{proof}

\end{document}